\documentclass[11pt]{article}
\usepackage[top=1in, bottom=1in, left=1in, right=1in, letterpaper]{geometry}

\usepackage[T1]{fontenc}

\usepackage[sc,osf]{mathpazo}   
\usepackage[euler-digits,small]{eulervm}

\usepackage{amsmath, amssymb, amsfonts, bm}
\usepackage{amsthm}
\usepackage{mathtools}
\usepackage{cite}

\usepackage{algorithm}
\usepackage[noend]{algpseudocode}
\usepackage{graphicx}
\usepackage{color}
\usepackage{enumitem}
\usepackage{xspace}
\usepackage{units}
\usepackage{tikz} 
\usepackage{fancybox}

\usepackage{bbm}

\usepackage[hyphens]{url}
\usepackage[colorlinks=true,citecolor=blue,linkcolor=blue,bookmarksnumbered=true,hyperfootnotes=true,breaklinks]{hyperref}
\usepackage[nameinlink]{cleveref}

\newtheorem{theorem}{Theorem}
\newtheorem{proposition}[theorem]{Proposition}

\newtheorem{lemma}[theorem]{Lemma}
\newtheorem{remark}[theorem]{Remark}

\theoremstyle{remark}
\newtheorem*{notation}{Notation}

\theoremstyle{definition}
\newtheorem{definition}[theorem]{Definition}
\newtheorem{fact}{Fact}

\crefname{theorem}{Theorem}{Theorems}
\crefname{lemma}{Lemma}{Lemmas}
\crefname{proposition}{Proposition}{Propositions}
\crefname{corollary}{Corollary}{Corollaries}

\crefname{fact}{Fact}{Facts}
\crefname{observation}{Observation}{Observations}
\crefname{claim}{Claim}{Claims}
\crefname{condition}{Condition}{Conditions}

\crefname{example}{Example}{Examples}
\crefname{definition}{Definition}{Definitions}
\crefname{remark}{Remark}{Remarks}
\crefname{note}{Note}{Notes}
\crefname{notation}{Notation}{Notations}
\crefname{section}{Section}{Sections}
\crefname{appendix}{Appendix}{Appendices}

\crefname{secnum}{Section}{Sections}
\crefname{enumi}{part}{parts}

\crefname{algorithm}{Algorithm}{Algorithms}

\newcommand\numberthis{\addtocounter{equation}{1}\tag{\theequation}}

\newcommand{\Reals}{\mathbb{R}}
\newcommand{\Complex}{\mathbb{C}}

\newcommand{\zo}{\{0,1\}}
\newcommand{\eps}{\varepsilon}
\newcommand{\half}{\tfrac{1}{2}} 

\newcommand{\ceq}{\coloneqq}

\DeclareMathOperator*{\E}{\mathbb{E}}

\DeclareMathOperator*{\Var}{Var}

\DeclarePairedDelimiter{\norm}{\lVert}{\rVert}

\newcommand{\dotp}{\,\tikz[baseline=-0.5ex]\draw[black,fill=black,radius=1pt] (0,0) circle ;\,}%
\newcommand{\ip}[2]{#1 \dotp #2}

\newcommand{\Sign}{\mathsf{Sign}}
\newcommand{\poly}{\mathsf{poly}}

\newcommand{\dsbs}{\mathsf{DSBS}}
\newcommand{\bgs}{\mathsf{BGS}}

\newcommand{\Ball}{\mathsf{Ball}}

\newcommand{\code}{\mathcal{C}}

\newcommand{\utail}{\mathcal{U}}
\newcommand{\bitail}{\mathcal{B}}
\newcommand{\tritail}{\mathcal{T}}

\newcommand{\Normal}{\mathcal{N}}

\newcommand{\Ex}{\mathbb E}

\newcommand{\wt}{\mathsf{wt}} 

\newcommand{\Bin}{{\rm Bin}}

\newcommand{\GF}{\mathbb{F}}

\newcommand{\cH}{\mathcal{H}}

\newcommand{\cV}{\mathcal{V}}

\newcommand{\prho}{\bar{\rho}}

\newcommand{\pub}{\ensuremath{\text{pub}}}

\newcommand{\IC}{\mathrm{IC}}
\newcommand{\ext}{\mathrm{ext}}
\newcommand{\inte}{\mathrm{int}}

\newcommand{\COMMENT}[1]{}

\title{Resource-Efficient Common Randomness and Secret-Key Schemes}

\author{
Badih Ghazi \thanks{Computer Science and Artificial Intelligence Laboratory, Massachusetts Institute of Technology, Cambridge MA 02139. {\tt badih@mit.edu}. Part of the work done while at IBM Research - Almaden.}
\and
T.S. Jayram \thanks{  IBM Research - Almaden. {\tt jayram@us.ibm.com}.}
}

\date{}

\begin{document}

\maketitle

\begin{abstract}

We study \emph{common randomness} where two parties have 
access to i.i.d.
samples from a known random source, and wish to generate
a shared random key using limited (or no) communication with
the largest possible probability of agreement. This
problem is at the core of secret key generation in cryptography,
with connections to communication under uncertainty
and locality sensitive hashing. 
We take the approach of treating correlated sources as a critical
resource, and ask whether common randomness can be generated
\emph{resource-efficiently}.

We consider two notable sources in this setup arising from correlated
bits and correlated Gaussians. We design the first \emph{explicit}
schemes that use only a \emph{polynomial} number of samples  
(in the key length) so that the players can generate shared keys
that agree with constant probability using optimal communication.
The best previously known schemes were both non-constructive and 
used an exponential number 
of samples. 
In the \emph{amortized} setting,
we characterize the largest achievable ratio of key length 
to communication in terms of the \emph{external} and
\emph{internal} information costs, two well-studied quantities in 
theoretical computer science.
In the relaxed setting where the two parties merely
wish to improve the \emph{correlation}
between the generated keys of length $k$, we show
that there are no interactive protocols using $o(k)$
bits of communication having agreement probability 
even as small as $2^{-o(k)}$.
For the related communication problem where the players
wish to compute a joint function $f$ of their inputs 
using i.i.d. samples from a known source,
we give a \emph{zero-communication} protocol using 
$2^{O(c)}$ bits where $c$ is the \emph{interactive}
randomized public-coin communication complexity of $f$.
This  matches the lower bound shown previously while
the best previously known upper bound was doubly exponential
in $c$.

Our schemes reveal a new connection between common randomness and \emph{unbiased error-correcting codes}, e.g., dual-BCH codes and 
their analogues in Euclidean space.

\end{abstract}

\thispagestyle{empty}

\newpage
\tableofcontents
\thispagestyle{empty}

\setcounter{page}{0}

\newpage
\section{Introduction}\label{sec:intro}

\emph{Common randomness} plays a fundamental role in various problems of cryptography
and information theory. We study this problem in the basic two-party communication
setting in which Alice and Bob wish to agree on a (random) key by 
drawing i.i.d. samples from a known source such as 
correlated bits or correlated Gaussians.
If we further require that an eavesdropper, upon seeing the communication only, 
gains no information about the shared key, then this defines a 
\emph{secret key} scheme. 
This information-theoretic approach to security was
introduced in the seminal works of Mauer~\cite{maurer1993secret} and 
Ahlswede and Csisz{\'a}r~\cite{ahlswede1993common}.
Both common randomness and secret-key generation 
have been extensively studied in information 
theory~\cite{ahlswede1998common,csiszar2000common,gacs1973common,Wyner_CommonInfo,
csiszar2004secrecy,zhao2011efficiency,tyagi2013common, liu2015secret,liu2016common}.
Common randomness  has applications to identification 
capacity~\cite{ahlswedeD89} and hardware-based procedures for extracting 
a unique random ID from process 
variations~\cite{lim2005extracting,su2008digital,yu2009towards}
that can be used in authentication~\cite{lim2005extracting,suh2007physical}.

Randomness is a powerful tool as well in the algorithm designer's arsenal.
Shared keys (aka public randomness) are used crucially in the design
of efficient communication protocols with immediate applications to
diverse problems in streaming, sketching, data structures and 
property testing.
Common randomness is thus a natural model for studying how shared keys can be generated
in settings where it is not available
directly~\cite{mossel2004coin,mossel2006non,bogdanov2011extracting,
chan2014extracting,CGMS_ISR,guruswami2016tight}.
In this paper, we take the approach of treating correlated sources as a critical
algorithmic resource, and ask whether common randomness can be generated
\emph{efficiently}.\footnote{%
Notably, the schemes that we design can also be easily transformed 
into secret key schemes, as shown later.} 

For $-1 \le \rho \le 1$,
we say that $(X,Y) \sim \dsbs(\rho)$ (\emph{doubly symmetric binary source})
if $X,Y$ are both uniform over $\{\pm1\}$ and 
their correlation (and covariance) $\E[XY] = \rho$
(i.e., a \emph{binary symmetric channel} with uniform input).
We say that $(X,Y) \sim \bgs(\rho)$ (\emph{bivariate Gaussian source})
if $X,Y \sim \Normal(0,1)$,
the standard normal distribution, and their correlation is again $\rho$. 

Bogdanov and Mossel~\cite{bogdanov2011extracting}
gave a common randomness scheme for $\dsbs(\rho)$
with zero-communication to generate 
$k$-bit keys that agree with probability 
$2^{-\frac{1-\rho}{1+\rho}\cdot k}$, up to lower order
inverse $\poly(k,1-\rho)$ factors (which we suppress henceforth).
Using the hypercontractive properties of the noise 
operator~\cite{bonami1970etude,beckner1975inequalities}, 
they also proved the ``converse'' that the bound for agreement (probability) 
is essentially the best possible.
In followup work, Guruswami and Radhakrishnan~\cite{guruswami2016tight} recently
gave a one-way scheme that achieves an \emph{optimal} tradeoff between
communication and agreement.\footnote{%
They also use hypercontractivity to prove the converse, which extends
to other sources including $\bgs(\rho)$.}
Note that a simple scheme in which Alice just sends her input 
requires $k - O_{\rho}(1)$ bits of communication for
constant agreement.
In contrast, their scheme can guarantee the same agreement 
using only $(1-\rho^2)\cdot k$ bits of communication.
This is a nontrivial \emph{amortized} bound since 
for $\rho > 0$,
the ratio of entropy 
to communication (=$1/(1-\rho^2)$)
is strictly bounded away from 1  as $k \to \infty$.
On the other hand, the above schemes are 
\emph{non-explicit} (i.e., proved using the probabilistic method) 
and use an \emph{exponential} number of samples in $k$.
Bogdanov and Mossel~\cite{bogdanov2011extracting} asked
whether an \emph{explicit} and \emph{efficient} scheme 
can be designed, motivating the definition below. 

We say that a common randomness scheme to generate $k$-bit keys 
(with $k$ as input)
is \emph{resource-efficient}, 
if it
(i)~is explicitly defined, 
(ii)~uses $\poly(k)$ samples,
(iii)~has constant agreement probability, and
(iv)~achieves an amortized ratio of entropy
to communication bounded away from 1.
We give the \emph{first} efficient scheme for correlated bits and Gaussians, answering the question of~\cite{bogdanov2011extracting}.

\begin{theorem}\label{thm:cr-efficient}
There exist resource-efficient one-way common randomness schemes for $\dsbs(\rho)$
and $\bgs(\rho)$ using $(1-\rho^2) \cdot k$ bits of communication.
For zero-communication, there exist explicit schemes for $\dsbs(\rho)$
and $\bgs(\rho)$ using $\poly(k)$ samples with agreement probability 
$2^{-\frac{1-\rho}{1+\rho}\cdot k}$, up to polynomial factors.
\end{theorem}

More generally, we obtain one-way schemes with \emph{optimal} 
tradeoff between communication and agreement, matching
\cite{guruswami2016tight}, while using only $\poly(k)$ samples.
Below is the formal statement.

\begin{theorem}\label{thm:cr-explicit}
Let $0 < \rho < 1$ and 
$0 \le \delta \le \sqrt{\tfrac{1-\rho}{1+\rho}}$ be arbitrary.
Set $\varphi = \rho + \delta\sqrt{1-\rho^2}$.
Then there exist explicit one-way common randomness schemes for $\dsbs(\rho)$
and $\bgs(\rho)$ using $\poly(k)$ samples such that:

\begin{enumerate}[topsep=3pt,partopsep=0pt,parsep=0pt,itemsep=3pt]


\item the entropy of the key is at least $k-o(k)$;\footnote{%
We follow~\cite{guruswami2016tight} who actually consider the 
\emph{min-entropy} of Alice's output, which is justifiable on technical grounds.}

\item the agreement probability is at least $2^{-\delta^2 k}$, up to polynomial factors; and 

\item the communication is $O((1-\varphi^2)\cdot k)$ bits.

\end{enumerate}
\end{theorem}

We point out that our schemes are resource efficient but
\emph{computationally inefficient}. 
One representative challenge that arises here is in decoding dual-BCH 
codes, which are an explicit algebraic family of error-correcting codes,
from a \emph{very large} number of errors.

The above schemes follow a template that generalizes the 
approach taken by~\cite{bogdanov2011extracting,guruswami2016tight}.
It relies on a carefully constructed codebook 
$\code \subseteq \Reals^n$ of size $2^k$, where $n$ is the number of samples.
Alice outputs the codeword in $\code$ with the largest projection
while Bob does the same on a subcode of $\code$ based on Alice's message. 
The analysis of the template reduces it to the problem
of obtaining good tail bounds on the joint distribution induced
by these projections.
For $\bgs(\rho)$, we use a codebook consisting of an explicitly
defined large family of nearly-orthogonal vectors in $\mathbb{R}^n$ 
due to Tao \cite{tao}, who showed their near-orthogonality property
using the Weil bound for curves. 
The novel part of the analysis involves getting precise
conditional probability tail bounds on trivariate Gaussians 
induced by the projections, whose covariance matrix has a special structure. 
Standard methods only give asymptotic bounds on such tails
which is inadequate in the low-communication regime. Here,
the best possible agreement is exponentially small in $k$.
Our analysis determines the exact constant in the exponent
by carefully evaluating the underlying triple integrals.

The resource-efficient scheme for  $\dsbs(\rho)$ is based
on Dual-BCH codes that can be seen as an $\GF_2$-analogue of 
Tao's construction.
The Weil bound for curves implies that dual-BCH codes are
``unbiased'', in the sense that any two distinct codewords are at distance
$\approx n/2$ (with $n$ being the block length)\footnote{For more on
unbiased codes, we refer the reader to the work of Kopparty and Saraf
\cite{kopparty2013local}.}. 
Analogous to the Gaussian case, the analysis involves 
getting precise bounds on the (conditional) tail probabilities of
various correlated binomial sums. 
Since $n=\poly(k)$, we
cannot handle these binomial sums using the (two-dimensional) Berry-Esseen
theorem, since the incurred additive error of $1/\sqrt{n}$ would overwhelm
target agreement. Moreover, crude concentration and anti-concentration
bounds cannot be used since they do not determine the exact constant 
in the exponent. We directly handle these correlated binomial
sums, which turns out to involve some tedious calculations related to the
binary entropy function.

\paragraph{Interactive Common Randomness and Information Complexity.}
Ahlswede and Csisz{\'a}r~\cite{ahlswede1993common,ahlswede1998common} 
studied common randomness in their seminal work using an \emph{amortized} 
communication model. They defined it as the maximum achievable ratio 
$a/c$, such that for every large enough number of samples $n$, Alice
and Bob can agree on a key of $a \cdot n$ bits using $c \cdot n$ bits of
communication, where the agreement
probability tends to $1$ (as $n$ tends to infinity). 
This more stringent linear relationship between the quantities
is not obeyed by our explicit schemes.
For one-way communication,
they characterized this ratio in terms of the \emph{Strong Data Processing Constant} 
of the source, which is intimately related to its
hypercontractive properties~\cite{ahlswede1976spreading,anantharam2013maximal}. 
More recently, Liu, Cuff and Verdu \cite{liu2015secret,liu2016common,liurate} extended 
this beyond one-way communication. In particular, \cite{liurate} derives the ``rate region''
for $r$-round amortized common randomness.

In this work, we show that $r$-round amortized common randomness can be
alternatively characterized in terms of two well-studied notions in
theoretical computer science: the \emph{internal} and \emph{external}
information costs of communication protocols. Recall that the internal
information cost~\cite{bjks04,barak2013compress} of a two-party 
randomized protocol is
the total amount of information that each of the two players
learns about the other player's input, whereas its external information
cost~\cite{cswy01} is the amount of information that an external
observer learns about the inputs (see~\cref{sec:inf_comp} for formal
definitions). These measures have been extensively studied within the
context of communication complexity. While being interesting measures in
their own rights, they have also been the central tool in tackling
direct-sum problems, with numerous applications, e.g., in data streams and
distributed computation.

\begin{theorem}[Informal Statement]\label{thm:inf_comp}
Given an arbitrary distribution $\mu$, let $\Gamma_r$ denote the supremum over
all $r$-round randomized protocols $P$ of the 
ratio of the external information cost
to the internal information cost of $P$ with respect to $\mu$. Then, for
$r$-round amortized common randomness, $\Gamma_r$ equals the largest achievable
ratio $H/R$ such that using $\mu$ as the source, for every large enough $n$, 
Alice and Bob, can agree on a key of $H\cdot n - O(1)$ bits with probability 
$1-o_n(1)$ using $r$ rounds and $R \cdot n + O(1)$ bits of communication.
\end{theorem}

For the proof, we use a direct-sum approach,
a classical staple of information complexity arguments.
Our setup is slightly different from the known direct-sum results because we 
need to lower bound the internal information cost of the $n$-input protocol
as well as upper bound its external information cost (which 
is non-standard) \emph{simultaneously}.
The essential ingredients are the same: embed the input on a 
judiciously chosen coordinate but the argument works an round-by-round basis so as to keep the mutual information expressions intact.
To prove the other direction, we use the rate region 
of~\cite{liu2016common,liurate}  to get a lower bound on $\Gamma_r$.

\smallskip

Finally, we outline various settings 
where common randomness plays an important role.
\begin{description}[wide=0\parindent]
\item[Secret Key Generation:] 
While secret key generation requires common randomness, 
in the amortized setting they are known to imply each 
other~\cite{liu2016common,liurate}:
the rate pair $(H,R)$, using the notation of~\cref{thm:inf_comp}, is 
achievable for common randomness if and only if $(H-R,R)$
is achievable for secret key generation. In particular,
using the Strong Data Processing Constant for $\dsbs(\rho)$, 
the rate ratio $H/R=1/(1-\rho^2)$ is achievable for common randomness
and the rate ratio $\rho^2/(1-\rho^2)$ for secret key generation, 
but using non-explicit schemes.
Our resource-efficient but non-amortized schemes given in~\cref{thm:cr-efficient} 
can be easily transformed into secret key schemes.
See~\cref{rem:secret-key}.

\item[General Sources:]
\Cref{thm:cr-explicit} also implies an explicit scheme for an 
\emph{arbitrary} source $\mu$ in terms of its \emph{maximal correlation} 
$\rho(\mu)$~\cite{hirschfeld1935connection,gebelein1941statistische,
renyi1959measures}. 
For $(X,Y) \sim \mu$, recall that $\rho(\mu) \ceq \sup \E F(X) G(Y)$
over all real-valued functions $F$ and $G$ with
$\E F(X) = \E G(Y) = 0$ and $\Var F(X)  = \Var G(Y)  =  1$.
This uses the idea (implicit in \cite{witsenhausen1975sequences}) that given 
i.i.d. samples from any source of maximal correlation $\rho$, there is a 
explicit strategy via CLT that allows Alice and Bob to use these samples 
in order to generate standard $\rho$-correlated Gaussians. 
The resulting scheme however is \emph{not} resource-efficient.


\item[Correlated Randomness Generation:]
In this relaxation proposed by~\cite{CGMS_ISR}, Alice and Bob are given
access to $\dsbs(\rho)$ and wish to generate $k$ bits that are 
jointly distributed i.i.d. according to $\dsbs(\rho')$ 
where $\rho < \rho'$? Note that the $\rho' = 1$ corresponds to the 
the common randomness setup studied above. We partially answer a question of
\cite{CGMS_ISR} that even a modest improvement in the correlation 
requires substantial communication. 
Let $\eps' \log(1/\epsilon') \ll \eps < \half$ be fixed.
We show that for Alice and Bob to produce $k$ samples 
according to $\dsbs(1-2\eps')$ 
using $\dsbs(1-2\eps)$ as the source requires $\Omega(\eps \cdot k)$ bits of communication (even 
for interactive protocols and even when the agreement probability is as small as $2^{-o(k)}$).
See~\cref{sec:gcr-append} for a detailed description.

\item[Communication with Imperfect Shared Randomness:]
In this framework~\cite{bavarian2014role,CGMS_ISR} 
(see also~\cite{ghazi2016communication}), 
Alice and Bob wish to compute a joint function of their inputs
and have access to i.i.d samples from a known source. 
For example, with $\dsbs(\rho)$ this setup
interpolates between the well-studied public randomness
($\rho=1$) and private randomness ($\rho=0$) models.
Communication complexity lower bounds for
imperfect shared randomness give one approach
to rule out low-communication common randomness schemes. In
particular,~\cite{bavarian2014role} exhibit a (partial) function whose
zero-communication complexity using $\dsbs(\rho)$ for all $\rho<1$ is
exponentially larger than the one using public randomness. We prove that this
separation is tight. We show a stronger result that every
function having \emph{interactive} communication $c$ bits using public
randomness has a zero-communication protocol with $2^{O(c)}$ bits using 
$\dsbs(\rho)$ for every $\rho<1$. This answers a question of Sudan \cite{madhu}. 
See~\cref{sec:smp-append} for a detailed description.

\item[Locality Sensitive Hashing (LSH):] 
A surprising ``universality'' feature
of our schemes (as well as previous ones)  
for $\dsbs(\rho)$ and $\bgs(\rho)$ using zero-communication
is that their definition is oblivious 
to $\rho$; only the analysis for every
fixed $\rho$ shows that they have near-optimal agreement.
This has a close resemblance to schemes used in LSH.
Indeed, we show that our common randomness scheme leads to
an improvement in the ``$\prho$-parameter''~\cite{IndykM98}
that governs one aspect of the performance of an LSH scheme. 
While this is mathematically interesting,
we caution the reader that this does not lead to better
nearest-neighbor data structures since the improvement is only
qualitatively better and our scheme is computationally inefficient.
See~\cref{sec:lsh}.
\end{description}

\paragraph{Organization.}
\Cref{sec:one-way} describes the template used for the one-way schemes
and sets up the structure of the analysis. 
\Cref{sec:bgs} and ~\cref{sec:dsbs} describe the schemes 
for $\bgs(\rho)$ and $\dsbs(\rho)$ and their analysis.
In~\Cref{sec:inf_comp}, we show the connection between amortized common 
randomness and information complexity.
In \Cref{sec:conc}, we conclude with some very intriguing open questions.

\subsection{Preliminaries}\label{sec:prelim}

\begin{notation}
For a tuple $U=(U_1, U_2, \dots, U_n)$,
let $U_i^j \ceq (U_i, U_{i+1}, \dots, U_j)$, when 
$1 \le i \le j \le n$, and empty otherwise;
we may drop the subscript when $i=1$.
For a distribution $\mu$, let $\mu^{\otimes n}$ be obtained by taking
i.i.d. samples $(X_1,Y_1), \dots,(X_n,Y_n)$ from $\mu$.
Abusing notation, we say that $(X^n,Y^n) \sim \mu^{\otimes n}$.
Let $\ip{}{}$ denote the standard inner product and let 
$\norm{\cdot}$ denote the Euclidean norm over $\Reals$.
For any positive integer $n$, let $[n] \ceq \{1,\dots,n\}$.
Let $a \lesssim b$ denote $a \le Cb$ for some positive global constant $C$.
\end{notation}

\smallskip

\noindent\textbf{Bivariate Gaussians.}
Let $(X,Y) \sim \bgs(\rho)$.
Let $Q(t) \ceq \Pr[X>t]$ denote the Gaussian tail probability
and $L(t,\varphi;\rho) \ceq \Pr\bigl[X > t, Y > \varphi t]$
denote the (asymmetric) orthant probability. In~\cref{sec:bgs-apdx},
we prove the following, which also uses some 
seemingly new properties of $Q(t)$.

\begin{proposition} \label{prop:QL-bound}
Let $t,\delta \ge 0$.
Set $\varphi \ceq \rho + \delta\sqrt{1-\rho^2}$ and 
$\lambda_0 \ceq \sqrt{\frac{2}{\pi}}$.
Then: 
\begin{alignat*}{2}
&(a) \ \ \frac{e^{-t^2/2}}{t+\lambda_0} \ \lesssim \ Q(t) 
\ \lesssim \ \frac{e^{-t^2/2}}{t+1/\lambda_0} \ \le \ e^{-t^2/2};
\qquad \qquad 
&&(b) \ \frac{Q(t)^{\delta^2}}{\delta t+\lambda_0} \ \lesssim \ 
Q(\delta t) \ \lesssim \ Q(t)^{\delta^2}(t+\lambda_0)^{c^2}; \\
&(c) \ \ L(t,\varphi; \rho) \ge \ Q(t) Q(\delta t);
\quad \text{ and } \qquad \qquad 
&&(d) \ \  Q(t) \ \le \ Q(\delta t) \ \le \ Q(t)^{\delta^2},
\quad \text{if \, $\delta \le 1$}
\end{alignat*}
\end{proposition}
\begin{proposition}[Elliptical symmetry]\label{prop:bgs-ellip}
For $v,w \in \Reals^n$ with unit norms,
$(\ip{v}{X},\ip{w}{Y}) \sim \bgs\bigl(\rho(\ip{v}{w})\bigr)$. 
\end{proposition}


\section{Template One-Way Scheme and its Analysis} \label{sec:one-way}

The one-way schemes (including zero-communication as a special case) 
have the following template.
Let $\mu$ denote the source on $\Reals \times \Reals$.
Alice and Bob will generate $n$ iid samples from $\mu$ and use them
to output $k$ bit keys.
This is achieved by the players using a special codebook 
$\code$ of $2^k$ points in $\Reals^n$ where each 
codeword has a $k$-bit encoding.
For $c \ge 1$, the players also agree on a coloring 
$\chi$ of $\code$ using $2^c$ colors
such that each color class has size at most $|\code| \cdot 2^{-c} + 1$.
In addition, let $\diamond$ denote an auxiliary color.
Thus, each color can be specified using $c+1$ bits.
For the special case of zero-communication, 
we assume wlog that all codewords are colored $\diamond$ and we set $c=0$.

Let $t$ and $s$ be parameters that govern the achievable min-entropy
and agreement probability.
Let $\kappa_A$ and $\kappa_B$ be mappings such that
$\kappa_A(X)$ and $\kappa_B(Y)$ are each 
uniformly distributed over $\zo^k$.
\begin{algorithm}
\caption{One-way scheme for source $\mu$}
\label{alg:one-way}
\begin{algorithmic}[1]

\Procedure{CR}{$k$;$\mu$}
\Comment{Generate $k$-bit common random key using source $\mu$.}

\State{Let $(X^n,Y^n) \sim \mu^{\otimes n}$. 
Let $X \ceq X^n$ and $Y \ceq Y^n$.}
\Comment{Alice gets $X$ and Bob gets $Y$.}

\If{$\exists$ unique $v \in \code$ such that $\ip{v}{X} > t$} 
    {Alice outputs $v$ and sends $\chi(v)$.}
\Else{} {Alice outputs $\kappa_A(X)$ 
           and sends $\diamond$.}
\EndIf

\State{Bob receives the color $\tau$.}

\If{$\exists$ unique 
    $w \in \code$ such that $\chi(w) = \tau$ and
    $\ip{v}{Y} > s$}
    {Bob outputs $w$}.
\Else{}
    {Bob outputs $\kappa_B(Y)$.}
\EndIf

\EndProcedure
\end{algorithmic}
\end{algorithm}

The pseudocode is given in~\cref{alg:one-way}.
For the analysis, define the following quantities:

\medskip

\noindent 1. \textbf{Univariate tail:}
$\utail \ceq \underset{v \in \code}{\max} \Pr[\ip{v}{X} > t]$;
\hspace{20pt}
2. \textbf{Bivariate tail:} 
$\bitail
\ceq \underset{v \in \code}{\min} \Pr[\ip{v}{X} > t, \ip{v}{Y} > s]$

\medskip

\noindent 3. \textbf{Conditional trivariate tails:}\\
\hspace*{0.5pt} (a) $\tritail_A \ceq \underset{v \ne w \in \code}{\max} 
      \Pr[\ip{w}{X} > t \mid \ip{v}{X} > t, \ip{v}{Y} > s]$
and 
(b) $\tritail_B \ceq \underset{v \ne w \in \code}{\max} 
      \Pr[\ip{w}{Y} > s \mid \ip{v}{X} > t, \ip{v}{Y} > s]$

\begin{theorem}\label{thm:ow-templ}
The min-entropy of the basic scheme is at least $-\log(\utail + 2^{-k})$.
Assume that $|\code| \cdot \tritail_A \le \tfrac{1}{4}$
and $|\code| \cdot \tritail_B \le \tfrac{1}{4}\cdot 2^c$.
Then the probability of agreement is at least 
$\half |\code| \cdot \bitail$.
\end{theorem}
\begin{proof}
If Alice outputs $a \in \zo^k$ then either there exists
a unique $v \in \code$ whose encoding is $a$. This happens
with probability at most $\Pr[\ip{v}{X} > t] \le \utail$. 
Otherwise, $\kappa_A(X) = a$ which happens with probability $2^{-k}$.
The min-entropy guarantee follows.  

For the agreement, fix $v \in \code$. Define event 
$E_v \ceq A_v \wedge \overline{B_v} \wedge \overline{C_v}$ where
$A_v \ceq \{\ip{v}{X} > t \wedge  \ip{v}{Y} > s\}$,
$B_v \ceq \{\exists w \ne v: \ip{w}{X} > t\}$, and
$C_v \ceq \{\exists w \ne v: \chi(v) = \chi(w) \wedge \ip{w}{Y} > s\}$.

Note that the event $E_v$ ensures that both players output the encoding
of $v$. 
By the union bound:
\begin{align*}
\Pr&[E_v] 
\ \ge \ \Pr[A_v] 
     \cdot \bigl(1-\Pr[B_v \vee C_v \mid A_v]\bigr) \\
&\ge \ \Pr[A_v] \cdot 
      \Bigl(1 - \sum_{w \ne v} \Pr[ \ip{w}{X} > t \mid A_v] 
       - \sum_{w \ne v}
       1\{\chi(w)=\chi(v)\}\cdot\Pr[ \ip{w}{Y} > s \mid A_v] \Bigr) \\
&\ge \ \bitail 
\bigl(1-|\code| \cdot \tritail_A-|\code| \cdot 2^{-c} \cdot \tritail_B\bigr)
\ \ge \ \half \bitail,
\end{align*}
where the last two inequalities follow from the definition 
of $\bitail$ and $\tritail$ 
and then invoking the premise of the lemma. 
Thus the agreement probability is at least 
$\sum_v \Pr[E_v] \ge \half |\code| \cdot \bitail$.
\end{proof}

As an illustration, we present an
explicit one-way scheme for the $\bgs(\rho)$
using an exponential number of samples.
Let $k$ be a large enough constant and let $n=2^k$.
Let $\code$ consist of the $n$
standard basis vectors $\bigl\{e_i : i \in [n]\bigr\}$ in $\Reals^n$. 
Choose $t>0$ so that the Gaussian tail probability 
$Q(t) = \tfrac{1}{4}\cdot 2^{-k}$.
Let $\rho \le \varphi \le 1$ be arbitrary and set $s = \varphi t$. 
(Choose $\varphi=1$ for zero-communication.)

For the analysis, 
note that for each $i$, we have $\ip{e_i}{X} = X_i$ and
$\ip{e_i}{Y} = Y_i$.
Therefore $\Pr[X_i > t] = Q(t)$ and 
so by~\cref{thm:ow-templ}, the min-entropy of Alice's output
is at least $-\log(Q(t) + 2^{-k}) \ge k-1$. 

We now analyze the agreement probability.
To bound the bivariate tail, first by~\cref{prop:QL-bound}\,(a),
we have $t = \Theta(\sqrt{k})$.
Let $\delta$ satisfy $\varphi = \rho + \delta\sqrt{1-\rho^2}$.
Observe that $0 \le \delta \le 1$.
Applying~\cref{prop:QL-bound}\,(b,c), we obtain:
\begin{equation}\label{eq:bitail}
\bitail \ = \ \min_{i \in [n]} \Pr[X_i > t, Y_i > \varphi t] = L(t,\varphi;\rho)
\ \gtrsim \ \frac{Q(t)^{1 + \delta^2}}{\delta t+\Theta(1)}
\ \gtrsim \ \frac{Q(t)^{1 + \delta^2}}{\delta\sqrt{k} + \Theta(1)}
\end{equation}

For $i \ne j$, 
the trivariate tail probability
$\Pr[X_j > t \mid  X_i > t, Y_i > \varphi t] = \Pr[X_j > t] = Q(t)$,
by independence of components of $(X,Y)$.
Similarly, 
$\Pr[Y_j > \varphi t \mid  X_i > t, Y_i > \varphi t] = Q(\varphi t)$.
Therefore:
\begin{equation}\label{eq:tritail-bgs}
\tritail_A \le Q(t) \qquad \text{and} \qquad \tritail_B \le Q(\varphi t)
\end{equation} 
Now $Q(t) = \tfrac{1}{4}\cdot 2^{-k}$, 
so $|\code| \cdot \tritail_A \le \tfrac{1}{4}$.
Next,
$Q(\varphi t) \le Q(t)^{\varphi^2}$, using ~\cref{prop:QL-bound}\,(d).
Therefore, $\tritail_B \le Q(t)^{\varphi^2}$.
If we choose $c \ge (1-\varphi^2)(k + 2)$, then it can be verified that
$|\code| \cdot \tritail_B \le \tfrac{1}{4}\cdot 2^c$. 
This ensures that the conditions of~\cref{thm:ow-templ} 
for agreement are satisfied.

By~\cref{thm:ow-templ}, the agreement probability is 
$\half |\code| \cdot \bitail
\gtrsim 2^{-\delta^2 k}/(\delta\sqrt{k} + \Theta(1))$
and the scheme uses $O((1-\varphi^2)k)$ bits of communication.
In particular,  set $\varphi=\rho$ and $\delta=0$; we obtain an explicit
one-way scheme with constant probability and $O((1-\rho^2)k)$ bits of 
communication.


\section{Efficient Scheme for $\bgs(\rho)$} \label{sec:bgs}

In this section, we give a resource-efficient one-way scheme for $\bgs(\rho)$ 
with optimal communication $(1-\rho^2)k$ bits.
More generally, the tradeoff between the communication and 
agreement probability is similar to the one obtained with the
scheme presented in~\ref{sec:one-way}.

The analysis of the template given previously suggests the following
scheme to reduce the sample complexity to $k=\poly(n)$:
use a codebook such that the projections are only 3-wise independent.
Unfortunately, this does not work since a multivariate Gaussian 
distribution is completely characterized by its first and second moments,  
so even pairwise independence would imply full independence! 
Instead, we use a codebook consisting of an explicitly
defined large family of nearly-orthogonal vectors in $\mathbb{R}^n$ 
due to Tao \cite{tao}, who showed their near-orthogonality property
using the Weil bound for curves. 

Let $p$ be a prime number and $n = 2 \cdot p$. 
We identify $\mathbb{R}^n$ with the complex vector space $\cV$ of functions 
from $\GF_p$ to $\Complex$, where $\Complex$ denotes the complex plane. 
Thus $v \in \cV$ will also denote an element of $\Reals^n$.
With this identification, we have 
$\ip{v}{w} = \mathfrak{Re} \bigl(\sum_{x \in \GF_p} v(x)\overline{w(x)}\bigr)$
for $v,w \in \cV$.

Let $d$ be a positive integer.
Let $\omega \ceq e^{2\pi i/p}$ denote the $p$-th root of unity.
For every $a \in \GF_p^d$, let $v_a \in \cV$ be defined as
$v_a(x) = \tfrac{1}{\sqrt{p}} \cdot \omega^{a_d x^d + \dots + a_1 x}$.
We set $\code \ceq \{ v_a \colon a \in \GF_p^d \}$. 
Note that the all elements of $\code$ have unit norm.
The Weil bound for curves then implies that for every 
$a \neq b \in \GF_p^d$, we have that 
$|\ip{v_a}{v_b}| \le (d-1)/\sqrt{p}$ \cite{weil1948courbes} (for a 
recent exposition see~\cite{kaufman2011new}).

Choose $d = o(n^{1/4}/\sqrt{\log{n}})$ and $k = d \cdot \log(n/2)$ in Tao's
construction. We use the same parameters $t$, $s$, $\varphi$ and $\delta$ for~\cref{alg:one-way} as in the previous scheme described in~\cref{sec:one-way}.

By elliptical symmetry~(\cref{prop:bgs-ellip}),
$(\ip{v}{X},\ip{v}{Y}) \sim \bgs(\rho)$, for every $v \in \code$. 
Therefore the bounds in~\cref{sec:one-way} 
for the univariate and bivariate tails (see~\cref{eq:bitail})
also hold here. 
The key difference is in the analysis of the trivariate probabilities
because we no longer have independence amongst the various pairs 
$(\ip{v}{X},\ip{v}{Y})$.
This requires a new analysis of the conditional tails involving 
trivariate Gaussians whose covariances have a special structure. 
Below, we show that a slightly weaker bound than~\cref{eq:tritail-bgs}:
$\tritail_A \le Q(t) \cdot (1+o_n(1))$  
and $\tritail_B \le Q(\varphi t) \cdot (1+o_n(1))$.
Nevertheless, we can apply the same argument following~\cref{eq:tritail-bgs}
in~\cref{sec:one-way} which implies again that: 
(a) the min-entropy at least $k-1$; 
(b) the agreement probability is $\gtrsim 2^{-\delta^2 k}/(\delta\sqrt{k} + \Theta(1))$; and
(c) the communication is $O((1-\varphi^2)k)$ bits.
In particular, with $\varphi=\rho$ and $\delta=0$; we obtain the main result of this 
section, namely a resource-efficient one-way scheme using  
$O((1-\rho^2)\cdot k)$ bits of communication.

It remains to prove that 
$\tritail_A \le Q(t) \cdot (1+o_n(1))$  
and $\tritail_B \le Q(\varphi t) \cdot (1+o_n(1))$.
Fix $v \ne w \in \code$.
The construction ensures that $|\ip{v}{w}| \le \theta$ with 
$\theta = (d-1)/\sqrt{p} = O(k/(\sqrt{n} \cdot \log{n}))$. 
For $k = o(n^{1/4} \cdot \sqrt{\log{n}})$, we have $\theta = o_n(1)$.

Now observe that $(\ip{w}{X}, \ip{v}{X}, \ip{v}{Y})$
can be written as a linear transform on $(X,Y)$, 
so jointly they have the trivariate Gaussian distribution,
Their joint distribution is fully given by the first two moments. 
By stability, the marginals are standard normal 
and by elliptical symmetry, the covariances 
can be calculated as
(i) $\E[(\ip{w}{X})(\ip{v}{X})] = \ip{w}{v} \le \theta$, 
(ii) $\E[(\ip{v}{X})(\ip{v}{Y})] = \rho$, and
(iii) $\E[(\ip{w}{X})(\ip{v}{Y})] = \rho (\ip{w}{v}) \le \rho\theta$.
Observe that $(\ip{w}{Y}, \ip{v}{Y}, \ip{v}{X})$ is also trivariate
with an identical mean and covariance matrix.

\begin{lemma}\label{lem:tritail-basic}
Let $(U,V,W)$ be a trivariate Gaussian with standard normal marginals 
and covariances $\E[UV] = \sigma$, $\E[VW] =\rho$,
and $\E[UW] = \sigma\rho$.
Let $r,r' \ge 0$. Then for all $b \ge 1$:
\[ \Pr[U > r \mid V > r, W > r'] \le  Q\biggl(\frac{1-b\sigma}{\sqrt{1-\sigma^2}}r\biggr) 
+ \frac{Q(br)}{\Pr[V > r, W > r']} \]
\end{lemma}
\begin{proof}
We have:
\begin{equation}\label{eq:bgs-trivar}
\Pr[U > r \mid V > r, W > r'] 
= \frac{\Pr[U > r, V > r, W > r']}{\Pr[V > r, W > r']}
\end{equation} 
For the numerator, we split the range of $V$ into two intervals:
\[ \Pr[U > r, V > r, W > r']
= \Pr[U > r, r < V \le br, W > r'] + \Pr[U > r, V > br, W > r'] \]
The second term is at most $\Pr[V > br] = Q(br)$.
For the first term, note that the covariance structure implies that
$U$ and $W$ are independent conditioned on $V$,
so we can write $U = \sigma V + \sqrt{1-\sigma^2}Z$.
where $Z \sim \Normal(0,1)$ is independent of $(V,W)$.
The event $\{U > r\}$ can be rewritten as 
$\bigl\{Z > \frac{r-\sigma V}{\sqrt{1-\sigma^2}}\bigr\}$ 
which under the assumption $\{V \le br\}$ implies that 
$\{Z > ar\}$ where $a \ceq \frac{1-b\sigma}{\sqrt{1-\sigma^2}}$.
By independence:
\[ \Pr[U > r, r < V \le br, W > r'] 
\le \Pr[Z > ar] \Pr[r < V \le br, W > r']
\le Q(ar) \Pr[V > r, W > r'] \]
Substituting these bounds in~\cref{eq:bgs-trivar} finishes the proof.
\end{proof}

Apply~\cref{lem:tritail-basic} to the triples 
$(\ip{w}{X},\ip{v}{X},\ip{v}{Y})$
with $r \ceq t$, $r' \ceq \varphi t$ and
$(\ip{w}{Y}, \ip{v}{Y}, \ip{v}{X})$
with $r \ceq \varphi t$, $r' \ceq t$.
In both cases, $\sigma \ceq \ip{v}{w} \le \theta$.
Since $Q(\cdot)$ is decreasing:
\begin{gather}
\Pr[\ip{w}{X} > t \mid \ip{v}{X} > t, \ip{v}{Y} > \varphi t]
  \ \le \ Q\biggl(\frac{1-b\theta}{\sqrt{1-\theta^2}}t\biggr) 
   + \frac{Q(bt)}{L(t,\varphi;\rho)}, \quad \forall b \ge 1  \label{eq:bgs-eff-a} \\
\Pr[\ip{w}{Y} > \varphi t \mid \ip{v}{Y} > \varphi t, \ip{v}{X} > t]
  \ \le \ Q\biggl(\frac{1-b\theta}{\sqrt{1-\theta^2}}\varphi t\biggr) 
   + \frac{Q(b\varphi t)}{L(t,\varphi;\rho)}, \quad \forall b \ge 1 \label{eq:bgs-eff-b}
\end{gather}
Set $b \ceq 2/\phi$. By~\cref{prop:QL-bound}\,(a), 
$Q(b \varphi t) \, \lesssim \, e^{-b^2 \varphi^2 t^2/2} = e^{-2t^2}$
and 
$Q(t) \, \gtrsim \, e^{-t^2/2}/(t + \lambda_0)$.
Because $t = \Theta(\sqrt{k})$, for large enough $k$, we have
$Q(b \varphi t) \, \lesssim \, Q(t)^3 e^{-t^2/2}(t+\lambda_0)^3 = Q(t)^3 o_n(1)$.

Using this bound and~\cref{prop:QL-bound}\,(c,d), we obtain:
\[ Q(bt) \le Q(b \varphi t) 
\le Q(t)^3 \cdot o_n(1) 
\le L(t,\varphi;\rho) Q(t) \cdot o_n(1) 
\le L(t,\varphi;\rho) Q(\varphi t) \cdot o_n(1) \]
Thus the second term in~\cref{eq:bgs-eff-a} (resp.~\cref{eq:bgs-eff-b}) 
is at most $Q(t)\cdot o_n(1)$ (resp. $Q(\varphi t)\cdot o_n(1)$).

For the first terms in the right side of~\cref{eq:bgs-eff-a} 
and~\cref{eq:bgs-eff-b},
let $a \ceq \frac{1-b\theta}{\sqrt{1-\theta^2}}$. Note that $a \le 1$.
Now $Q(at) \le Q(t)^{a^2}$ by~\cref{prop:QL-bound}\,(d).
We calculate $1-a^2 = \bigl(\frac{2b - (1+b^2)\theta}{1-\theta^2}\bigr)\theta
\le 4b\theta$, since $\theta \ll 2b/(1+b^2)$.
For the choice of $d$ we have $kb\theta = o_n(1)$.
Thus $Q(at)/Q(t) \le Q(t)^{a^2-1} \lesssim 2^{k(1-a^2)} \le 2^{4kb\theta} 
= 2^{o_n(1)} = 1+o_n(1)$.

By~\cref{lem:Q-tail-prop}, $Q(at)/Q(t)$ is increasing in $t$,
so $Q(a\varphi t)/Q(\varphi t) \le Q(at)/Q(t) \le Q(at)/Q(t) = 1+o_n(1)$.
Thus the first term in~\cref{eq:bgs-eff-a} (resp.~\cref{eq:bgs-eff-b}) 
is at most $Q(t)\cdot (1+o_n(1))$ (resp. $Q(\varphi t)\cdot (1+o_n(1))$).
Combine the above bounds for the two terms in~\cref{eq:bgs-eff-a,eq:bgs-eff-b}
to complete the analysis.
This completes the proofs of \Cref{thm:cr-efficient} and \Cref{thm:cr-explicit} for the $\bgs$ source.

\begin{remark}\label{rem:secret-key}
We modify the above resource-efficient scheme that uses $c=(1-\rho^2)\cdot k$ 
bits of communication to generate secret keys.
Assume wlog that codewords within the same color class are encoded
with the same prefix of $c$ bits. 
Now Alice just outputs the $k-c = \rho^2 \cdot k$-bit
suffix of her output.
We briefly sketch the analysis as follows.
Using the min-entropy property as well as a similar 
lower bound on the probability that Alice 
outputs a particular key (which essentially follows from the
same bounds on bivariate tails used above) it can be shown that
the communicated bits are nearly uniform as well and that
the suffix of the output is nearly uncorrelated with the prefix.
This ensures the secrecy of the key from the eavesdropper.
\end{remark}

\section{Efficient Scheme for $\dsbs(\rho)$} \label{sec:dsbs}

We give a resource-efficient one-way scheme for $\dsbs(\rho)$ 
with optimal communication $(1-\rho^2)\cdot k$ using the template
of~\cref{alg:one-way}.
It is based on dual-BCH codes which can be seen as finite field 
analogues of the nearly-orthogonal vectors used in~\cref{sec:bgs}.
It is more natural here but still equivalent to work with
$\zo^n$ instead of $\{\pm 1\}^n$ and the Hamming distance $\Delta$
instead of inner-product.

Let $\code_{dBCH} = \code_{dBCH(d,m)}$ be the dual-BCH code with parameters
$m = \log(n+1)$ and $d$ being any polynomial in $n$ that satisfies $d =
o(n^{1/4}/\sqrt{\log{n}})$. Then, $|\code_{dBCH}| = 2^k$ where 
$k = d \cdot \log(n+1)$ is a polynomial in $n$. 
Let $\code$ be an arbitrary subset of
$\code_{dBCH}$ of size $2^{k'} = 2^k/(\gamma \cdot n)$ where $\gamma >0$ is
a sufficiently large absolute constant to be chosen later on. We denote
$\code = \{ v_a : a \in \{0,1\}^{k'} \}$. We set 
$r \triangleq n/2 - t\sqrt{n}/2$ where $t>0$ satisfies
$Q(t) = (1/4) \cdot 2^{-k}$. 
Similar to before, let $\rho \le \varphi \le 1$ 
so that the communication is $O((1-\varphi^2) k)$ bits.
Recall that $\delta$ satisfies $\varphi = \rho + \delta\sqrt{1-\rho^2}$.

In the following, we prove
the appropriate uni-, bi- and trivariate tail bounds for the scheme.
These are stated in~\cref{prop:bin_tail_ub,lem:prob_lb_dsbs,lem:sec_cond_w}.
 The proof follows the same structure that was used 
for $\bgs(\rho)$.
It requires some bounds on binomial sums proved in~\cref{sec:binom-apdx}.
By incorporating them into~\cref{thm:ow-templ}, we obtain the desired
performance of the scheme.
Let $(X,Y) \sim \dsbs(\rho)^{\otimes n}$.

\begin{proposition}\label{prop:bin_tail_ub}
For any $u \in \Reals$ (possibly depending on $n$),
$\Pr[ |\wt(X) - n/2| \ge  u \sqrt{n}/2] \le \poly(n) \cdot Q(u)$.
\end{proposition}

\begin{lemma}\label{lem:prob_lb_dsbs}
For every $a \in \{0,1\}^{k'}$:
$\Pr[\Delta(v_a, X) \le r,  \Delta(v_a, Y) \le r'] \geq \frac{1}{\Theta(n^2)} \cdot 2^{-k} \cdot 2^{-k \cdot \delta^2}$.
\end{lemma}
\begin{proof}
Follows from~\cref{prop:bin_tail_lb} and~\cref{prop:first_cond_w}.
\end{proof}

\begin{lemma}\label{lem:sec_cond_w}
Let $v, v' \in \{0,1\}^n$ satisfy
$|\Delta(v,v') -  n/2| \le \theta \cdot n/2$,
where $\theta = O(k/(\sqrt{n} \cdot \log{n}))$. Then:
\begin{align*}
\Pr[\Delta(v', X) \le r \mid \Delta(v,X) \le r, \Delta(v,Y) \le r'] &\le O(n) \cdot Q(t) \\
\Pr[\Delta(v', Y) \le r' \mid \Delta(v,X) \le r, \Delta(v,Y) \le r'] &\le O(n) \cdot Q((1-\varphi)t)
\end{align*}
\end{lemma}
\begin{proof}
Let $\ell \triangleq n/2 -\theta \cdot n/2$. 
Without loss of generality, we assume that $v = 0^n$ is 
the all-zeros vector and that $v' = 1^{\ell} 0^{n-\ell}$. Then,
\begin{align*}
&\Pr[\Delta(v', X) \le r \mid \Delta(v,X) \le r, \Delta(v,Y) \le r']\\ 
&= \Pr[\Delta(v', X) \le r \mid \wt(X) \le r, \wt(Y) \le r']\\ 
&=\frac{\Pr[\Delta(v', X) \le r, \wt(X) \le r, \wt(Y) \le r']}{\Pr[\wt(X) \le r, \wt(Y) \le r']}\\ 
&= \frac{1}{\Pr[\wt(X) \le r, \wt(Y) \le r']} \cdot \displaystyle\sum\limits_{r_1 = 0}^r \displaystyle\sum\limits_{r_2 = 0}^r \displaystyle\sum\limits_{r_3 = 0}^{r'} \Pr[\Delta(v', X) = r_1, \wt(X) = r_2, \wt(Y) = r_3]\\ 
&= \frac{1}{\Pr[\wt(X) \le r, \wt(Y) \le r']}  \biggl(\displaystyle\sum\limits_{r_2 = 0}^r \displaystyle\sum\limits_{r_3 = 0}^{r'} \Pr[\wt(X) = r_2, \wt(Y) = r_3] \\
&\hspace{6cm} \cdot\displaystyle\sum\limits_{r_1 = 0}^r \Pr[\Delta(v', X) = r_1 \mid \wt(X) = r_2, \wt(Y) = r_3]\biggr)\\ 
&= \frac{1}{\Pr[\wt(X) \le r, \wt(Y) \le r']}  \cdot \displaystyle\sum\limits_{r_2 = 0}^r \displaystyle\sum\limits_{r_3 = 0}^{r'} \Pr[\wt(X) = r_2, \wt(Y) = r_3] \cdot \displaystyle\sum\limits_{r_1 = 0}^r \Pr[\Delta(v', X) = r_1 \mid \wt(X) = r_2]\\ 
&= \frac{1}{\Pr[\wt(X) \le r, \wt(Y) \le r']}  \cdot \displaystyle\sum\limits_{r_2 = 0}^r \displaystyle\sum\limits_{r_3 = 0}^{r'} \Pr[\wt(X) = r_2, \wt(Y) = r_3] \cdot \Pr[\Delta(v', X) \le r \mid \wt(X) = r_2],
\numberthis \label{al:disc_cond}
\end{align*}
where the penultimate equality follows from the fact that $\Delta(v', X) - \wt(X) - \wt(Y)$ is a Markov chain.

For every non-negative integer $t_2$ satisfying $t_2 = o(n^{1/4})$ and $\theta \cdot t \cdot t_2 = o_n(1)$, we have that
\begin{align*}
\Pr&[\Delta(v',X) \le r \mid \wt(X) = n/2 - t_2 \sqrt{n}/2] \\
&= \displaystyle\sum\limits_{a=0}^{a_{max}} \psi(a)\\ 
&\overset{(A)}{\le} (a_{max}+1) \cdot \psi(a_{max})\\ 
&\overset{(B)}{\le} O(n) \cdot \Theta\bigg(\frac{1}{\sqrt{n}}\bigg) \cdot e^{- \frac{t^2}{2}}\\ 
&\overset{(C)}{\le} O(n) \cdot Q(t),
\numberthis \label{al:gd_t_2}
\end{align*}
where~$(A)$ follows from~\cref{prop:monot}, 
~$(B)$ from~\cref{prop:psi_a_max} and the fact that $\theta = o_n(1)$, 
and~$(C)$ from~\cref{prop:QL-bound}\,(a)
and the facts that $t = \Theta(\sqrt{k})$ and $k \le n$. 
Note that by assumption 
$\theta = O(k/(\sqrt{n} \cdot \log{n}))$. Thus, for any 
$k = o(n^{1/4} \cdot \sqrt{\log{n}})$, there exists a function  
$\nu(t,\theta) = \omega_n(1)$ satisfying 
$\nu(t,\theta) = o_n(\min(n^{1/4}, 1/(t \cdot \theta)))$ and
\begin{equation}\label{eq:nu_disc}
\Theta(n^{2}) \cdot 2^{k+k \cdot \delta^2} \cdot \exp(-\nu(t,\theta)^2) \le Q(t).
\end{equation}
We fix such a function $\nu(t,\theta)$ and set $\tau(t,\theta) \triangleq n/2 - \nu(t,\theta) \sqrt{n}/2$.
\Cref{al:disc_cond} now becomes:
\begin{equation}\label{al:alpha_beta_dsbs}
\Pr[\Delta(v', X) \le r \mid \Delta(v,X) \le r, \Delta(v,Y) \le r']
= \frac{1}{\Pr[\wt(X) \le r, \wt(Y) \le r']}  \cdot (\alpha + \beta),
\end{equation}
where
\begin{equation*}
\alpha \triangleq \displaystyle\sum\limits_{r_2 = \tau(t,\theta)}^r \displaystyle\sum\limits_{r_3 = 0}^{r'} \Pr[\wt(X) = r_2, \wt(Y) = r_3] \cdot \Pr[\Delta(v', X) \le r \mid \wt(X) = r_2],
\end{equation*}
and
\begin{equation*}
\beta \triangleq \displaystyle\sum\limits_{r_2 = 0}^{\tau(t,\theta)} \displaystyle\sum\limits_{r_3 = 0}^{r'} \Pr[\wt(X) = r_2, \wt(Y) = r_3] \cdot \Pr[\Delta(v', X) \le r \mid \wt(X) = r_2].
\end{equation*}
Using~\cref{al:gd_t_2} and the fact that $\nu(t,\theta) = o(\min(n^{1/4}, 1/(t \cdot \theta)))$, we get that
\begin{equation}\label{al:alpha_ub_dsbs}
\alpha \le \sum_{r_2 = \tau(t,\theta)}^r \sum_{r_3 = 0}^{r'} \Pr[\wt(X) = r_2, \wt(Y) = r_3] \cdot O(n) \cdot Q(t) 
\le O(n) \cdot Q(t) \cdot \Pr[\wt(X) \le r, \wt(Y) \le r'].
\end{equation}
We also have that
\begin{equation}\label{al:beta_ub_dsbs}
\beta \le \sum_{r_2 = 0}^{\tau(t,\theta)} \sum_{r_3 = 0}^{r'} \Pr[\wt(X) = r_2, \wt(Y) = r_3]
\le \Pr[\wt(X) \le \tau(t,\theta)]
\le \exp(-\nu(t,\theta)^2),
\end{equation}
where the last inequality uses the fact that $\nu(t,\theta) = \omega_n(1)$ and follows from~\cref{prop:bin_tail_ub} and~\cref{prop:QL-bound}\,(a). Combining~\cref{al:alpha_beta_dsbs},~\cref{al:alpha_ub_dsbs} and~\cref{al:beta_ub_dsbs}, we get that
\begin{align*}
\Pr[\Delta(v', X) \le r \mid \Delta(v,X) \le r, \Delta(v,Y) \le r'] &\le O(n) \cdot Q(t) + \frac{\exp(-\nu(t,\theta)^2)}{\Pr[\wt(X) \le r, \wt(Y) \le r']}\\ 
&\le O(n) \cdot Q(t) + \Theta(n^{2}) \cdot 2^k \cdot 2^{-k \cdot \delta^2} \cdot \exp(-\nu(t,\theta)^2)\\ 
&\le O(n) \cdot Q(t),
\end{align*}
where the second inequality follows from~\cref{prop:bin_tail_lb} and~\cref{prop:first_cond_w}, and the third inequality follows from the fact that $\nu(t,\theta)$ satisfies~\cref{eq:nu_disc}. This completes the proof of the first part of the lemma. The proof of the second part follows along the same lines.
\end{proof}

We note that the above bounds imply the desired result for agreement
probabilities up to $1/\poly(k)$. The result also holds for
\emph{constant} agreement probability. The main idea is to
combine the constant agreement scheme for the Gaussian source 
along with a multi-dimensional Berry-Esseen Theorem
(e.g., Theorem 67 of \cite{matulef2010testing}).\footnote{Since we are
dealing with constant error probabilities, the additive error from the
Berry-Esseen theorem is negligible.} The details are deferred to a future
version.


\section{Information Complexity and Common Randomness}\label{sec:inf_comp}
In this section, we show an intimate relationship between the achievable regions for amortized common randomness generation and the \emph{internal} and \emph{external} information costs of communication protocols,
two well-studied notions in theoretical computer science.
We say that $(H,R_1,R_2)$ is \emph{$r$-achievable} for a distribution
$\mu$ if for every $\eps>0$ there exists an $r$-round common randomness scheme $\Pi$ with 
$(X^n,Y^n) \sim \mu^{\otimes n}$  as inputs, for some $n=n(\eps)$,
where $n \to \infty$ as $\eps \to 0$,
such that the following holds:
let $M_t$ denote the message sent in round $t$ in $\Pi$,
and let $K_A$ (resp. $K_B$) denote the output of Alice (resp. Bob).
Then (1) $\sum_{t\ \text{odd}} H(M_t) \le (R_1+\eps)n$, 
(2) $\sum_{t\ \text{even}} H(M_t) \le (R_2+\eps)n$,
(3) $H(K_A),H(K_B) \ge (H-\eps)n$ and
(4) $K_A$ and $K_B$ both belong to a domain of size
$cn$ for some absolute constant $c$ independent of $\eps$ and $n$.
(The min-entropy guarantee in our basic definition is stronger
than the combination of parts~(3) and~(4).)
 
\begin{definition}
Let $P$ be a two-player randomized communication protocol
with both public and private coins and
let $R_\pub$ denote the public randomness.
With a slight abuse in notation, given $(X,Y) \sim \mu$,  
let $P$ also denote the transcript of the protocol
on input $(X,Y)$.  Define the following
measures for the protocol with respect to $\mu$:
(i) the \emph{external information cost} $\IC^{\ext}(P)$ equals 
$I(X,Y; P \mid R_\pub)$;
(ii) the \emph{marginal internal information cost} $\IC^{\inte}_A(P)$ for Alice  equals
$I(X; P \mid Y R_\pub)$ and analogously 
$\IC^{\inte}_B(P) = I(Y; P \mid X R_\pub)$ for Bob.
The (total) \emph{internal information cost} equals
the sum of the two marginal costs.
\end{definition}
We now characterize the achievable region 
for a fixed source distribution $\mu$
in terms of internal and external information costs
of protocols with respect to $\mu$.

\paragraph{Converse.}
We extend the ideas 
present in several works, e.g.~\cite{kaspi85,ahlswede1998common,liu2016common}.
We need the following direct-sum property (\cref{lem:direct-sum} below) for information costs of randomized protocols that we crucially use in our analysis. This property differs from the known direct-sum results in that it simultaneously bounds the internal and external information costs of the single-coordinate protocol.
Its proof uses the following tool.
\begin{proposition}[{\cite[Lemma 4.1]{ahlswede1998common}}] \label{prop:chain}
Let $S,T,X^n,Y^n$ be arbitrary random variables. Then:
\[ I(X^n ; S \mid T) - I(Y^n ; S \mid T)
= \sum_{j=1}^n I(X_j ; S \mid X^{j-1} Y_{j+1}^n T) 
   - I(Y_j ; S \mid X^{j-1} Y_{j+1}^n T). \]
\end{proposition}
\begin{proof}
We have by telescoping:
\begin{equation}\label{eq:telescoping}
I(X^n ; S \mid T) - I(Y^n ; S \mid T)
= \sum_{j=1}^n I(X^j Y_{j+1}^n ; S \mid T) - I(X^{j-1} Y_j^n ; S \mid T).
\end{equation}

By the chain rule for mutual information, for each $j \in [n]$, we have that
\[ I(X^j Y_{j+1}^n ; S \mid T) 
= I(X^{j-1} Y_{j+1}^n ; S \mid T) + I(X_j ; S \mid X^{j-1} Y_{j+1}^n T) \]
and 
\[ I(X^{j-1} Y_j^n ; S \mid T) 
= I(X^{j-1} Y_{j+1}^n ; S \mid T) + I(Y_j ; S \mid X^{j-1} Y_{j+1}^n T) \]
The proposition now follows by substituting the last two equations in \Cref{eq:telescoping}.
\end{proof}

\begin{lemma}[Direct sum]\label{lem:direct-sum}
Fix a distribution $\mu$ and an $r$-round randomized protocol $\Pi$ with
inputs $(X^n,Y^n)$ $\sim \mu^{\otimes n}$. Then there exists
an $r$-round randomized protocol $P$ with inputs $(X,Y) \sim \mu$ such that
(a) $\IC^{\inte}_A(\Pi) = n \cdot \IC^{\inte}_A(P)$,
(b) $\IC^{\inte}_B(\Pi) = n \cdot \IC^{\inte}_B(P)$,
and (c) $\IC^{\ext}(\Pi) \le n \cdot \IC^{\ext}(P)$.
\end{lemma}
\begin{proof}
For ease of presentation we suppress the public randomness of $\Pi$ in the expressions appearing in the proof below. Let $M_t$ be the message sent in $\Pi$ during round $t \in [r]$;
set $M_{r+1} \ceq \emptyset$.
We will be using the following properties of $\Pi$:
\begin{enumerate}[label=\Roman*.,ref=(\Roman*)]
\item For every odd $t \le r$, 
$I(Y^n; M_t \mid X^n M^{t-1}) = I(X^n; M_{t+1} \mid Y^n M^t) = 0$. 
\label{enum:all-rd}
\item For all $j \in [n]$ and odd $t \le r$,
$I(Y_j; M_t \mid X^j Y_{j+1}^n M^{t-1}) 
= I(X_j; M_{t+1} \mid X^{j-1} Y_j^n M^t) = 0$. 
This can also be shown, see, 
e.g.,~\cite[Eqns. 3.10--3.13]{kaspi85}. \label{enum:one-rd}
\end{enumerate}

We present the argument for the marginal internal information
cost for Alice;
a similar argument can be carried out for Bob's case as well.
Observe that:
\begin{equation} \label{eq:ic-ds-basic}
\IC^{\inte}_A(\Pi) = I(X^n; M^r \mid Y^n) =
\sum_{t \le r} I(X^n; M_t \mid Y^n M^{t-1}) 
= \sum_{t \ \text{odd}} I(X^n; M_t \mid Y^n M^{t-1}),
\end{equation}
by~\cref{enum:all-rd} above.
Fix an odd $t$ in the above sum. Again by~\cref{enum:all-rd} above:
\begin{equation} \label{eq:ic-y-indep}
I(X^n; M_t \mid M^{t-1}) 
= I(X^nY^n; M_t \mid M^{t-1}) 
= I(Y^n; M_t \mid M^{t-1}) + I(X^n; M_t \mid Y^n M^{t-1}),
\end{equation}
and therefore,
\begin{align}\label{eq:ic-odd-i}
I(X^n; M_t &\mid Y^n M^{t-1}) 
= I(X^n; M_t \mid M^{t-1}) - I(Y^n; M_t \mid M^{t-1})\nonumber\\
&\overset{(a)}{=} 
\sum_{j=1}^n I(X_j; M_t \mid X^{j-1} Y_{j+1}^n M^{t-1}) 
               - I(Y_j; M_t \mid X^{j-1} Y_{j+1}^n M^{t-1})\nonumber\\
&= \sum_{j=1}^n I(X_j; M_t \mid Y_j X^{j-1} Y_{j+1}^n M^{t-1})\nonumber\\ 
&= \sum_{j=1}^n I(X_j; M_t M_{t+1} \mid Y_j X^{j-1} Y_{j+1}^n M^{t-1})
\end{align}
where~(a) follows from~\cref{prop:chain}, and each of the last two equalities follows from the chain rule and by invoking \cref{enum:one-rd}.
We now substitute \cref{eq:ic-odd-i} in~\cref{eq:ic-ds-basic},
and sum over all odd $t$. 
\begin{equation} \label{eq:ic-alice}
\begin{split}
\IC^{\inte}_A(\Pi)
& = I(X^n; M^r \mid Y^n)
=  \sum_{t \ \text{odd}} \sum_{j=1}^n I(X_j; M_t M_{t+1} \mid  Y_j X^{j-1} Y_{j+1}^n M^{t-1}) \\
&= \sum_{j=1}^n I(X_j; M^r \mid  Y_j X^{j-1} Y_{j+1}^n) 
= n \cdot I(X_J; M^r \mid  Y_J X^{J-1} Y_{J+1}^n J),
\end{split}
\end{equation}
using the chain rule and then defining $J$ to be uniform over $[n]$ 
and independent of all the other random variables. 
Similarly for Bob:
\begin{equation} \label{eq:ic-bob}
\IC^{\inte}_B(\Pi)
= n \cdot I(Y_J; M^r \mid  X_J X^{J-1} Y_{J+1}^n J).
\end{equation}
We claim that the right side of~\cref{eq:ic-alice,eq:ic-bob}
are respectively the marginal internal information
costs for Alice and Bob in some protocol $P$ with inputs $(X,Y) \sim \mu$. Specifically, on input pair $(X,Y)$, the protocol $P$ simulates the protocol $\Pi$ by settting $X_J \ceq X$ and $Y_J \ceq Y$, and associating the public randomness with $J$, $X^{J-1}$, and $Y_{J+1}^n$.
\Cref{enum:one-rd} above ensures that
the messages in protocol $P$
can be generated by the players using private randomness.

It remains to bound the external information cost of $P$.
Observe that $\IC^{\ext}(P)$ equals
\begin{equation}\label{eq:ic-ext-cost}
I(X_J,Y_J; M^r \mid  X^{J-1} Y_{J+1}^n J)
= I(Y_J; M^r \mid  X^{J-1} Y_{J+1}^n J)
+ I(X_J ; M^r \mid  Y_J X^{J-1} Y_{J+1}^n J).
\end{equation}
The second term in~\cref{eq:ic-ext-cost} above equals
$\tfrac{1}{n}\cdot I(X^n; M^r \mid Y^n)$ via~\cref{eq:ic-alice}. 
For the first term,
using the independence of coordinates, 
\[ I(Y_J; M^r \mid  X^{J-1} Y_{J+1}^n J)
= I(Y_J; M^r X^{J-1} \mid  Y_{J+1}^n J)
\ge I(Y_J; M^r \mid  Y_{J+1}^n J) 
= \tfrac{1}{n} \cdot I(Y^n; M^r),\]
where we expand over $J$ and use the chain rule.
Combining the bounds for the two terms, we conclude:
\[ n \cdot \IC^{\ext}(P) \ge I(Y^n; M^r) + I(X^n; M^r \mid Y^n) 
= I(X^n Y^n; M^r) = \IC^{\ext}(\Pi).  \qedhere\]
\end{proof}

\begin{theorem} \label{thm:ic-converse}
If a tuple $(H,R_1,R_2)$ is $r$-achievable 
then for every $\eps>0$ there exists a randomized $r$-round protocol 
whose marginal internal information cost for Alice (resp. Bob)
with respect to the distribution $\mu$ is at most 
$R_1 + O(\eps)  + 1/n$ (resp. $R_2 + O(\eps) + 1/n$)
and whose external information cost is at least $H - \eps$.
\end{theorem}

\begin{proof}
Fix $\eps>0$. Let $n$ be such that there is an $r$-round protocol for common randomness generation
$\Pi$ on inputs $(X^n,Y^n) \sim \mu^{\otimes n}$.
Let $M_t$ denote the message sent in round $t$ in $\Pi$.
Let $K_A$ (resp. $K_B$) denote the output of Alice (resp. Bob).
We have (1) $\sum_{t\ \text{odd}} H(M_t) \le (R_1+\eps)n$, 
(2) $\sum_{t\ \text{even}} H(M_t) \le (R_2+\eps)n$,
(3) $H(K_A),H(K_B) \le (H-\eps)n$ and
(4) $K_A$ and $K_B$ both belong to a domain of size
$cn$ for some absolute constant $c$ (independent of $\epsilon$ and $n$).

Consider the case where $r$ is odd (the other case can be
handled similarly) and define a new protocol $\Pi'$
where Alice also sends $K_A$ to Bob along with the last message.
The number of rounds is still $r$.
Applying \Cref{lem:direct-sum}, there 
exists an $r$-round randomized protocol $P$ with inputs 
$(X,Y) \sim \mu$ such that 
$\IC^{\inte}_A(\Pi') = n \cdot \IC^{\inte}_A(P)$
and $\IC^{\inte}_B(\Pi') = n \cdot \IC^{\inte}_B(P)$.
Now since $\Pi'$ depends only on $X^n$ and $Y^n$, we have that $\IC^{\inte}_A(\Pi') = I(X^n; M^r K_A \mid Y^n)
= I(X^n; M^r \mid Y^n) + I(X^n; K_A \mid Y^n M^r)$.
Because $X^n \perp M_t \mid Y^n M^{t-1}$ for each even 
round $t$, by the chain rule, the first term equals 
\[ \sum_t I(X^n; M_t \mid Y^nM^{t-1})
= \sum_{t\ \text{odd}} I(X^n; M_t \mid Y^nM^{t-1})
\le \sum_{t\ \text{odd}} H(M_t) \le (R_1+\eps)n. \]
The second term is at most $H(K_A \mid Y^n M^r)$.
Now $K_B$ is determined by $Y^n$ and $M^r$ and
$\Pr[K_A \ne K_B] \le \eps$, so by Fano's inequality,
$H(K_A \mid Y^n M^r) \le \eps cn + 1$.
Therefore, 
$\IC^{\inte}_A(P) \le R_1 + \eps(1 + c) + 1/n$.
For Bob, the analysis is similar and even simpler
because his messages are unchanged (from $\Pi$ to $\Pi'$) so
$\IC^{\inte}_B(P) \le R_2 + \eps$.
(The bound stated in the lemma
is weaker because Fano's inequality is used 
when $r$ is even.)
Finally, apply \Cref{lem:direct-sum} to bound the external information 
cost of $P$ as 
\[ n \cdot \IC^{\ext}(P) \ge \IC^{\ext}(\Pi') 
= I(X^n Y^n; M^r K_A)
= H(M^r K_A) - H(M^r K_A \mid X^n Y^n) 
= H(M^r K_A). \] 
But $H(M^r K_A) \ge H(K_A) \ge (H-\eps) n$, so the
desired bound follows.
\end{proof}

\paragraph{Achievability.}
In~\cite{liu2016common}, a sufficient condition 
using Markov chains on auxiliary random variables
is given the existence of an interactive
common randomness scheme.
To fulfill this condition, their construction uses a random encoding 
argument.
We connect these conditions to the existence of an $r$-round 
communication protocol with the appropriate information costs. 

\begin{proposition}[\cite{liu2016common}]\label{prop:lcv}
Let $(X,Y) \sim \mu$.
Suppose there exist auxiliary random variables 
$U_1, U_2, \dots, U_r$ for some $r$ in some joint probability
space with $X$ and $Y$ where the marginal distribution
of $(X,Y)$ is $\mu$ satisfying the following:
\begin{enumerate}
\item For every odd $t$, $Y \perp U_t \mid X U^{t-1}$
and for every even $t$, $X \perp U_{t+1} \mid  Y U^t$. \label{enum:markov}

\item $\sum_{t\ \text{odd}} I(X; U_t \mid U^{t-1}) 
+ \sum_{t\ \text{even}} I(Y; U_t \mid U^{t-1}) \ge H$.
\label{enum:entropy}

\item $\sum_{t\ \text{odd}} I(X; U_t \mid U^{t-1}) 
- \sum_{t\ \text{odd}} I(Y; U_t \mid U^{t-1}) \le R_1$.
\label{enum:rateA}

\item $\sum_{t\ \text{even}} I(Y; U_t \mid U^{t-1}) 
- \sum_{t\ \text{even}} I(X; U_t \mid U^{t-1}) \le R_2$.
\label{enum:rateB}
\end{enumerate}
Then, there exists an $r$-round interactive
common randomness generation scheme $\Pi(X^n,Y^n)$
using $n$ i.i.d. samples as input where Alice sends 
at most $R_1 n$ bits, Bob sends at most $R_2 n$ bits
and the entropy of their output is at least $Hn$ bits
where the agreement probability tends to 1 as $n \to \infty$.
\end{proposition}

\begin{theorem} \label{thm:direct-converse}
	If there exists a $r$-round randomized protocol with inputs $(X,Y) \sim \mu$ whose marginal internal information cost for Alice (resp. Bob) is at most $R_1$ (resp. $R_2$) and whose external information cost is at least $H$, then $(H,R_1,R_2)$ is $r$-achievable.
\end{theorem}
\begin{proof}
Let $P$ be a randomized protocol with inputs $(X,Y) \sim \mu$ whose marginal internal information cost for Alice (resp. Bob) is at most $R_1$ (resp. $R_2$) and whose external information cost is at least $H$. Without loss of generality, we assume that $P$ uses no public randomness. For every $t \in [r]$, we let $U_t$ denote the message sent in $P$ during round $r$. We claim that the $U_t$'s satisfy the conditions in \Cref{prop:lcv}. First, note that the conditional independencies given 
in~\cref{enum:markov} of \Cref{prop:lcv} are equivalent to the
message structure of an $r$-round randomized protocol, and are thus satisfied by the $U_t$'s.

For every odd $t$, by~\cref{enum:markov}, 
$I(Y; U_t \mid X U^{t-1}) = 0$, so
\begin{align*}
I(X; U_t \mid U^{t-1}) = I(XY; U_t \mid U^{t-1}) 
= I(Y; U_t \mid U^{t-1}) + I(X; U_t \mid Y U^{t-1}).
\end{align*}
Therefore, 
$I(X; U_t \mid Y U^{t-1})
= I(X; U_t \mid U^{t-1}) - I(Y; U_t \mid U^{t-1})$.
By the chain rule,
\[ \IC^{\inte}_A(P) = I(X; U^r \mid Y)
= \sum_t I(X; U_t \mid Y U^{t-1})
= \sum_{t\ \text{odd}} I(X; U_t \mid U^{t-1}) - I(Y; U_t \mid U^{t-1}) 
\le R_1, \] 
via~\cref{enum:markov} where we used 
$I(X; U_t \mid Y U^{t-1})=0$
for every even $t$. Using the given assumption that $\IC^{\inte}_A(P) \le R_1$, we deduce that the $U_t$'s satisfy \cref{enum:rateA} of \Cref{prop:lcv}.
A similar argument using the given assumption that $\IC^{\inte}_B(P)  \le R_2$ implies that the $U_t$'s satisfy \cref{enum:rateB} of \Cref{prop:lcv}.

Applying a similar reasoning, we also obtain that:
\begin{align*}
\sum_{t\ \text{odd}} I(X; U_t \mid U^{t-1}) 
+ \sum_{t\ \text{even}} I(Y; U_t \mid U^{t-1}) &=  \sum_{t\ \text{odd}} I(XY; U_t \mid U^{t-1}) 
+ \sum_{t\ \text{even}} I(XY; U_t \mid U^{t-1}) \\
&=  I(XY; U^r)\\ 
&= \IC^{\ext}(P).
\end{align*}
The given assumption that $\IC^{\ext}(P) \geq H$ now implies that the $U_t$'s satisfy \cref{enum:entropy} of \Cref{prop:lcv}. Therefore, we conclude that $(H,R_1,R_2)$ is $r$-achievable.
\end{proof}

Combining \Cref{thm:ic-converse} and
\Cref{thm:direct-converse}, we 
obtain the the formal version of \Cref{thm:inf_comp}.
\begin{theorem}
Let $\Gamma_r$ denote
the supremum over all $r$-round randomized protocols $\Pi$ of 
the ratio of the external information cost to the internal information cost
of $\Pi$ with respect to $\mu$.
Then, $\Gamma_r$ equals the supremum of $H/(R_1+R_2)$
such that $(H,R_1,R_2)$ is $r$-achievable for $\mu$.
\end{theorem}

%
%

\section{Conclusion and Open Questions}\label{sec:conc}

The most important open question raised in this work is to obtain
\emph{computationally} efficient schemes for common randomness.
In particular, is there a resource-efficient scheme that 
also has time complexity $\poly(k)$? For our schemes, 
it not at all clear how to implement the decoding phase
time-efficiently (either over $\mathbb{F}_2$ or in Euclidean space). In
fact, even the slightly sub-exponential time algorithm
of~\cite{kopparty2013local} for decoding dual-BCH codes falls short of
working for the error radii that are needed to achieve near-optimal
agreement probability!

The sample complexity $n = o(k^4)$ of our explicit schemes is polynomial
but still far from the linear non-explicit sample schemes arising from
amortized common randomness.
The Kabatjanskii-Levenstein bound (cf.~\cite{tao}) implies that 
no nearly-orthogonal families of vectors (including the one we used) 
will achieve a linear sample complexity in our setup. 
Can we rule out linear sample schemes altogether?
One challenge is that such a proof cannot
solely rely on hypercontractivity because they ``tensorize'' 
and are thus oblivious to the number $n$ of used
samples. 

Our one-way scheme for general sources with \emph{maximal correlation} 
$\rho$ is explicit but not sample-efficient because it uses the CLT to 
reduce the problem to $\bgs(\rho)$.
Moreover, the tradeoff between communication and agreement is stated
in terms of $\rho$, whereas the best known negative results are in 
terms of \emph{hypercontractivity}. \cite{anantharam2013maximal} give an 
example of a source separating its maximal correlation from its
\emph{Strong Data Processing Constant}, 
which is intimately related to its hypercontractive properties. 
Can such a source be used to prove
that the tradeoff stated in~\Cref{thm:cr-explicit} is not tight 
for general sources?

A characterization of amortized \emph{correlated} 
randomness would be interesting even for one-way as
it would generalize the notion of the Strong Data Processing Constant.

Finally, our paper shows that tools used in common randomness 
could also be useful for Locality Sensitive Hashing. 
Can one establish a formal connection between these two areas?

\section*{Acknowledgements}\label{sec:ack}

The authors would like to thank Venkat Guruswami, Cl\'ement Canonne, Jingbo Liu, Ilya Razenshteyn and Madhu Sudan for very helpful discussions and pointers.

\bibliographystyle{alpha}
\bibliography{references}

\newcommand{\etalchar}[1]{$^{#1}$}
\begin{thebibliography}{MOR{\etalchar{+}}06}

\bibitem[AC93]{ahlswede1993common}
Rudolf Ahlswede and Imre Csisz{\'a}r.
\newblock Common randomness in information theory and cryptography. part {I}:
  Secret sharing.
\newblock {\em IEEE Transactions on Information Theory}, 39(4), 1993.

\bibitem[AC98]{ahlswede1998common}
Rudolf Ahlswede and Imre Csisz{\'a}r.
\newblock Common randomness in information theory and cryptography. {II}. {CR}
  capacity.
\newblock {\em IEEE Transactions on Information Theory}, 44(1):225--240, 1998.

\bibitem[AD89]{ahlswedeD89}
Rudolf Ahlswede and Gunter Dueck.
\newblock Identification via channels.
\newblock {\em {IEEE} Trans. Information Theory}, 35(1):15--29, 1989.

\bibitem[AG76]{ahlswede1976spreading}
Rudolf Ahlswede and Peter G{\'a}cs.
\newblock Spreading of sets in product spaces and hypercontraction of the
  markov operator.
\newblock {\em The annals of probability}, pages 925--939, 1976.

\bibitem[AGKN13]{anantharam2013maximal}
Venkat Anantharam, Amin Gohari, Sudeep Kamath, and Chandra Nair.
\newblock On maximal correlation, hypercontractivity, and the data processing
  inequality studied by {E}rkip and {C}over.
\newblock {\em arXiv preprint arXiv:1304.6133}, 2013.

\bibitem[BBCR13]{barak2013compress}
Boaz Barak, Mark Braverman, Xi~Chen, and Anup Rao.
\newblock How to compress interactive communication.
\newblock {\em SIAM Journal on Computing}, 42(3):1327--1363, 2013.

\bibitem[Bec75]{beckner1975inequalities}
William Beckner.
\newblock Inequalities in {F}ourier analysis.
\newblock {\em Annals of Mathematics}, pages 159--182, 1975.

\bibitem[BGI14]{bavarian2014role}
Mohammad Bavarian, Dmitry Gavinsky, and Tsuyoshi Ito.
\newblock On the role of shared randomness in simultaneous communication.
\newblock In {\em Automata, Languages, and Programming}, pages 150--162.
  Springer, 2014.

\bibitem[BJKS04]{bjks04}
Ziv {Bar-Yossef}, T.S. Jayram, Ravi Kumar, and D.~Sivakumar.
\newblock An information statistics approach to data stream and communication
  complexity.
\newblock {\em Journal of Computer and System Sciences}, 68(4):702--732, 2004.

\bibitem[BM11]{bogdanov2011extracting}
Andrej Bogdanov and Elchanan Mossel.
\newblock On extracting common random bits from correlated sources.
\newblock {\em Information Theory, IEEE Transactions on}, 57(10):6351--6355,
  2011.

\bibitem[Bon70]{bonami1970etude}
Aline Bonami.
\newblock {\'E}tude des coefficients de {F}ourier des fonctions de ${L}^p(g)$.
\newblock In {\em Annales de l'institut Fourier}, volume~20, pages 335--402,
  1970.

\bibitem[CGMS14]{CGMS_ISR}
Clement Canonne, Venkat Guruswami, Raghu Meka, and Madhu Sudan.
\newblock Communication with imperfectly shared randomness.
\newblock {\em ITCS}, 2014.

\bibitem[CMN14]{chan2014extracting}
Siu~On Chan, Elchanan Mossel, and Joe Neeman.
\newblock On extracting common random bits from correlated sources on large
  alphabets.
\newblock {\em Information Theory, IEEE Transactions on}, 60(3):1630--1637,
  2014.

\bibitem[CN00]{csiszar2000common}
Imre Csisz{\'a}r and Prakash Narayan.
\newblock Common randomness and secret key generation with a helper.
\newblock {\em IEEE Transactions on Information Theory}, 46(2):344--366, 2000.

\bibitem[CN04]{csiszar2004secrecy}
Imre Csisz{\'a}r and Prakash Narayan.
\newblock Secrecy capacities for multiple terminals.
\newblock {\em IEEE Transactions on Information Theory}, 50(12):3047--3061,
  2004.

\bibitem[CSWY01]{cswy01}
Amit Chakrabarti, Yaoyun Shi, Anthony Wirth, and Andrew Yao.
\newblock Informational complexity and the direct sum problem for simultaneous
  message complexity.
\newblock In {\em Foundations of Computer Science, 2001. Proceedings. 42nd IEEE
  Symposium on}, pages 270--278. IEEE, 2001.

\bibitem[Due10]{duembgen2010bounding}
Lutz Duembgen.
\newblock Bounding standard gaussian tail probabilities.
\newblock {\em arXiv preprint arXiv:1012.2063}, 2010.

\bibitem[Geb41]{gebelein1941statistische}
Hans Gebelein.
\newblock Das statistische problem der korrelation als variations-und
  eigenwertproblem und sein zusammenhang mit der ausgleichsrechnung.
\newblock {\em ZAMM-Journal of Applied Mathematics and Mechanics/Zeitschrift
  f{\"u}r Angewandte Mathematik und Mechanik}, 21(6):364--379, 1941.

\bibitem[GK73]{gacs1973common}
Peter G{\'a}cs and J{\'a}nos K{\"o}rner.
\newblock Common information is far less than mutual information.
\newblock {\em Problems of Control and Information Theory}, 2(2):149--162,
  1973.

\bibitem[GKS16]{ghazi2016communication}
Badih Ghazi, Pritish Kamath, and Madhu Sudan.
\newblock Communication complexity of permutation-invariant functions.
\newblock In {\em Proceedings of the Twenty-Seventh Annual ACM-SIAM Symposium
  on Discrete Algorithms}, pages 1902--1921. SIAM, 2016.

\bibitem[GR16]{guruswami2016tight}
Venkatesan Guruswami and Jaikumar Radhakrishnan.
\newblock Tight bounds for communication-assisted agreement distillation.
\newblock In {\em 31st Conference on Computational Complexity, {CCC} 2016, May
  29 to June 1, 2016, Tokyo, Japan}, pages 6:1--6:17, 2016.

\bibitem[Hir35]{hirschfeld1935connection}
Hermann~O Hirschfeld.
\newblock A connection between correlation and contingency.
\newblock In {\em Mathematical Proceedings of the Cambridge Philosophical
  Society}, volume~31, pages 520--524. Cambridge Univ Press, 1935.

\bibitem[IM98]{IndykM98}
Piotr Indyk and Rajeev Motwani.
\newblock Approximate nearest neighbors: Towards removing the curse of
  dimensionality.
\newblock In {\em Proceedings of the Thirtieth Annual {ACM} Symposium on the
  Theory of Computing, Dallas, Texas, USA, May 23-26, 1998}, pages 604--613,
  1998.

\bibitem[Kas85]{kaspi85}
Amiram~H. Kaspi.
\newblock Two-way source coding with a fidelity criterion.
\newblock {\em {IEEE} Trans. Information Theory}, 31(6):735--740, 1985.

\bibitem[KL11]{kaufman2011new}
Tali Kaufman and Shachar Lovett.
\newblock New extension of the weil bound for character sums with applications
  to coding.
\newblock In {\em Foundations of Computer Science (FOCS), 2011 IEEE 52nd Annual
  Symposium on}, pages 788--796. IEEE, 2011.

\bibitem[KS13]{kopparty2013local}
Swastik Kopparty and Shubhangi Saraf.
\newblock Local list-decoding and testing of random linear codes from high
  error.
\newblock {\em SIAM Journal on Computing}, 42(3):1302--1326, 2013.

\bibitem[LCV15]{liu2015secret}
Jingbo Liu, Paul Cuff, and Sergio Verd{\'u}.
\newblock Secret key generation with one communicator and a one-shot converse
  via hypercontractivity.
\newblock In {\em 2015 IEEE International Symposium on Information Theory
  (ISIT)}, pages 710--714. IEEE, 2015.

\bibitem[LCV16]{liu2016common}
Jingbo Liu, Paul~W. Cuff, and Sergio Verd{\'{u}}.
\newblock Common randomness and key generation with limited interaction.
\newblock {\em CoRR}, abs/1601.00899, 2016.

\bibitem[Liu16]{liurate}
Jingbo Liu.
\newblock Rate region for interactive key generation and common randomness
  generation.
\newblock {\em Manuscript available at
  \url{http://www.princeton.edu/~jingbo/preprints/RateRegionInteractiveKeyGen120415.pdf}
  (visited on 02/13/2017)}, 2016.

\bibitem[LLG{\etalchar{+}}05]{lim2005extracting}
Daihyun Lim, Jae~W Lee, Blaise Gassend, G~Edward Suh, Marten Van~Dijk, and
  Srinivas Devadas.
\newblock Extracting secret keys from integrated circuits.
\newblock {\em IEEE Transactions on Very Large Scale Integration (VLSI)
  Systems}, 13(10):1200--1205, 2005.

\bibitem[Mau93]{maurer1993secret}
Ueli~M Maurer.
\newblock Secret key agreement by public discussion from common information.
\newblock {\em IEEE Transactions on Information Theory}, 39(3):733--742, 1993.

\bibitem[MO04]{mossel2004coin}
Elchanan Mossel and Ryan O'Donnell.
\newblock Coin flipping from a cosmic source: On error correction of truly
  random bits.
\newblock {\em arXiv preprint math/0406504}, 2004.

\bibitem[MOR{\etalchar{+}}06]{mossel2006non}
Elchanan Mossel, Ryan O'Donnell, Oded Regev, Jeffrey~E Steif, and Benny
  Sudakov.
\newblock Non-interactive correlation distillation, inhomogeneous markov
  chains, and the reverse {B}onami-{B}eckner inequality.
\newblock {\em Israel Journal of Mathematics}, 154(1):299--336, 2006.

\bibitem[MORS10]{matulef2010testing}
Kevin Matulef, Ryan O'Donnell, Ronitt Rubinfeld, and Rocco~A Servedio.
\newblock Testing halfspaces.
\newblock {\em SIAM Journal on Computing}, 39(5):2004--2047, 2010.

\bibitem[R{\'e}n59]{renyi1959measures}
Alfr{\'e}d R{\'e}nyi.
\newblock On measures of dependence.
\newblock {\em Acta mathematica hungarica}, 10(3-4):441--451, 1959.

\bibitem[SD07]{suh2007physical}
G~Edward Suh and Srinivas Devadas.
\newblock Physical unclonable functions for device authentication and secret
  key generation.
\newblock In {\em Proceedings of the 44th annual Design Automation Conference},
  pages 9--14. ACM, 2007.

\bibitem[She99]{sheppard1899application}
WF~Sheppard.
\newblock On the application of the theory of error to cases of normal
  distribution and normal correlation.
\newblock {\em Philosophical Transactions of the Royal Society of London.
  Series A, Containing Papers of a Mathematical or Physical Character}, pages
  101--531, 1899.

\bibitem[SHO08]{su2008digital}
Ying Su, Jeremy Holleman, and Brian~P Otis.
\newblock A digital 1.6 pj/bit chip identification circuit using process
  variations.
\newblock {\em IEEE Journal of Solid-State Circuits}, 43(1):69--77, 2008.

\bibitem[Sud14]{madhu}
Madhu Sudan.
\newblock Personal communication.
\newblock 2014.

\bibitem[Tao13]{tao}
Terence Tao.
\newblock A cheap version of the {K}abatjanskii-{L}evenstein bound for almost
  orthogonal vectors.
\newblock
  \url{https://terrytao.wordpress.com/2013/07/18/a-cheap-version-of-the-kabatjanskii-levenstein-bound-for-almost-orthogonal-vectors/},
  2013.

\bibitem[Tya13]{tyagi2013common}
Himanshu Tyagi.
\newblock Common information and secret key capacity.
\newblock {\em IEEE Transactions on Information Theory}, 59(9):5627--5640,
  2013.

\bibitem[Wei48]{weil1948courbes}
Andr{\'e} Weil.
\newblock {\em Sur les courbes alg{\'e}briques et les vari{\'e}t{\'e}s qui s'
  en d{\'e}duisent}.
\newblock Number 1041. Hermann, 1948.

\bibitem[Wit75]{witsenhausen1975sequences}
Hans~S Witsenhausen.
\newblock On sequences of pairs of dependent random variables.
\newblock {\em SIAM Journal on Applied Mathematics}, 28(1):100--113, 1975.

\bibitem[Wyn75]{Wyner_CommonInfo}
Aaron~D. Wyner.
\newblock The common information of two dependent random variables.
\newblock {\em {IEEE} Transactions on Information Theory}, 21(2):163--179,
  1975.

\bibitem[YLH{\etalchar{+}}09]{yu2009towards}
Haile Yu, Philip Heng~Wai Leong, Heiko Hinkelmann, L~Moller, Manfred Glesner,
  and Peter Zipf.
\newblock Towards a unique {FPGA}-based identification circuit using process
  variations.
\newblock In {\em 2009 International Conference on Field Programmable Logic and
  Applications}, pages 397--402. IEEE, 2009.

\bibitem[ZC11]{zhao2011efficiency}
Lei Zhao and Yeow-Kiang Chia.
\newblock The efficiency of common randomness generation.
\newblock In {\em 2011 49th Annual Allerton Conference on Communication,
  Control, and Computing (Allerton)}, 2011.

\end{thebibliography}

\appendix

\newpage
\section{Properties of Bivariate Gaussian Distribution}
\label{sec:bgs-apdx}

\begin{proposition}[Elliptical symmetry,~\cref{prop:bgs-ellip} restated]
Let $(X,Y) \sim \bgs(\rho)^{\otimes n}$ and $v,w \in \Reals^n$ have unit norm.
$(\ip{v}{X},\ip{w}{Y}) \sim \bgs\bigl(\rho(\ip{v}{w})\bigr)$. 
\end{proposition}
\begin{proof}
Since $(\ip{v}{X},\ip{w}{Y})$ is a linear transform of $(X,Y)$ it
has the bivariate Gaussian distribution.
Thus, we only need to determine the first and second moments. 
Since $v$ and $w$ have unit-norm, 
by stability, the marginals are standard normal. Finally,
we verify that their covariance is $\rho(\ip{v}{w})$.
\[ \E[\ip{v}{X} \ip{w}{Y}]
= \sum_{i,j=1}^n  v(i) w(j) \cdot \E[X_i \cdot Y_j]
= \rho \sum_{i=1}^n  v(i) w(i) = \rho (\ip{v}{w}) \qedhere \]
\end{proof}

\subsection{Tail bounds for Gaussians}

The following bounds are well-known; using 
Duembgen's approach~\cite{duembgen2010bounding}, 
we prove them below to make it self-contained. 
Let $\lambda(t) \ceq \tfrac{\phi(t)}{Q(t)}$ denote the \emph{inverse Mills ratio},
i.e the ratio of the density function to the tail probability of 
a standard normal random variable.
Let $\lambda_0 \ceq \lambda(0) = \sqrt{\tfrac{2}{\pi}}$.
\begin{lemma} \label{lem:mills-bound}
For all $t \ge 0$,
$\max\{t, \lambda_0^2 \cdot t + \lambda_0\} \ \le \ \lambda(t)
\ \le \ t + \min\{1/t,\lambda_0\}$.
Equality holds only at $t=0$.
\end{lemma}
\begin{proof}
For all $t \ge 0$ and any function 
$\alpha \colon \Reals_+ \to \Reals_+$ such that
$\lim_{t\to\infty} \alpha(t) = \infty$, 
let
\[ f_{\alpha}(t) \ceq \frac{\phi(t)}{\alpha(t)} - Q(t), \] 
so that $\lim_{t \to \infty} f_{\alpha}(t) = 0$.
Observing that $Q'(t) = -\phi(t)$ and $\phi'(t) = -t\phi(t)$, we have: 
\[ \frac{\partial f_{\alpha}}{\partial t} 
= \frac{\phi(t)}{\alpha(t)^2} 
   \bigl(\alpha(t)^2 - t \cdot \alpha(t) - \alpha'(t)\bigr) \]
Thus, the sign of the partial derivative is determined by 
$g_{\alpha}(t) \ceq \alpha(t)^2 - t \cdot \alpha(t) - \alpha'(t)$.

\begin{enumerate}
\item $\alpha(t) = t + 1/t$: In this case,
$g_{\alpha}(t) =  2/t^2 > 0$.
Therefore, $f_{\alpha}(t)$ is strictly increasing in $t$; together
with $f_{\alpha}(0) = -\half$ 
and $\lim_{t \to \infty} f_{\alpha}(t) = 0$,
it follows that 
that $f_{\alpha}(t) < 0$ for all $t \ge 0$.

\item $\alpha(t) = t + \lambda_0$: In this case,
$g_{\alpha}(t) =  \lambda_0 t + \lambda_0^2 - 1$ is linear in $t$.
Set $d \ceq (1 - \lambda_0^2)/\lambda_0 > 0$, and it follows that
$g_{\alpha}(t) < 0$ for $0 \le t < d$
and $g_{\alpha}(t) > 0$ for $t > d$.
Therefore, $f_{\alpha}(t)$ is decreasing in $t$ over $[0,d]$ 
and increasing in $t$ over $[d,\infty)$; the endpoint conditions
imply that $f_{\alpha}(t) \le 0$ for all $t \ge 0$ with equality
only at $t=0$.

\item $\alpha(t) = t$: In this case, $g_{\alpha}(t) = -1$
so $f_{\alpha}(t)$ is strictly decreasing in $t$.
Now $\lim_{t \to 0} f_{\alpha}(t) = \infty$ therefore
$f_{\alpha}(t) > 0$ for all $t \ge 0$.

\item $\alpha(t) = \lambda_0^2 \cdot t + \lambda_0$: In this case,
$g_{\alpha}(t)$ 
is quadratic in $t$ with a zero constant term.
Set $d \ceq \tfrac{2 \lambda_0^2 - 1}{\lambda_0(1-\lambda_0)} > 0$ and 
an easy calculation shows that 
$g_{\alpha}(t) > 0$ for $0 \le t < d$
and $g_{\alpha}(t) \le 0$ for $t > d$.
An analogous argument implies that $f_{\alpha}(t) \ge 0$ 
for all $t \ge 0$ with equality only at $t=0$. 
\end{enumerate}
\end{proof}

We now show two interesting properties of the tail probability function. 
These seem to be new as far as we know.
\begin{lemma} \label{lem:Q-tail-prop}
The function $Q(t)^{1/t^2}$ 
is increasing in $t$ for $t \ge 0$.
For every fixed $0 \le a \le 1$, the function $Q(at)/Q(t)$ is 
increasing in $t$ for $t \ge 0$.
\end{lemma}
\begin{proof}
We use the basic identities $(\ln Q(t))' = -\lambda(t)$ and 
$\lambda'(t) = \lambda(t)^2 - t \lambda(t)$.

For the first property, it suffices to show that
the function $f(t) \ceq \tfrac{1}{t^2} \cdot \ln Q(t)$
is increasing in $t$ for $t \ge 0$.
We have:
\[ \frac{df}{dt} = -\frac{t\lambda(t)+2\ln(Q(t))}{t^3}.\]
Let $u(t) \ceq t\lambda(t) + 2\ln(Q(t))$ and observe that
$u'(t) = \lambda(t)(t \cdot \lambda(t) - t^2 - 1) < 0$ 
by~\cref{lem:mills-bound}.
Thus $f'(t) > 0$ and $f(t)$ is increasing in $t$.

For the second property, it suffices to show that the function
$g(t,a) \ceq \ln Q(at) - \ln Q(t)$ for each fixed $0 \le a \le 1$
is increasing in $t$ for $t \ge 0$.
We have:
\[ \frac{\partial g}{\partial t} = \lambda(t) - a \cdot \lambda(at) \]
At $t=0$ the right side equals 0 and for $t>0$ we will show that 
$\lambda(t) > a \cdot \lambda(bt)$. This would imply the desired property
that $g(t,a)$ is increasing in $t$.
Multiplying both sides by $t$, we need to show that
$t\cdot\lambda(t) > at \cdot \lambda(at)$, that is, 
the function $h(x) \ceq x \cdot \lambda(x)$ is an increasing function
of $x$ for $x \ge 0$.
This holds because $h'(x) = \lambda(x)(1-x^2 + x \lambda(x)) > 0$
by~\cref{lem:mills-bound}.
\end{proof}

We are ready to prove~\cref{prop:QL-bound}.
\begin{proposition}[\cref{prop:QL-bound} restated]
Let $t,\delta \ge 0$.
Set $\eta \ceq \rho + \delta\sqrt{1-\rho^2}$ and 
$\lambda_0 \ceq \sqrt{\frac{2}{\pi}}$.
Then: 
\begin{alignat*}{2}
&(a) \ \ \frac{e^{-t^2/2}}{t+\lambda_0} \ \lesssim \ Q(t) 
\ \lesssim \ \frac{e^{-t^2/2}}{t+1/\lambda_0} \ \le \ e^{-t^2/2};
\qquad \qquad 
&&(b) \ \frac{Q(t)^{\delta^2}}{\delta t+\lambda_0} \ \lesssim \ 
Q(\delta t) \ \lesssim \ Q(t)^{\delta^2}(t+\lambda_0)^{c^2}; \\
&(c) \ \ L(t,\eta; \rho) \ge \ Q(t) Q(\delta t);
\quad \text{ and } \qquad \qquad 
&&(d) \ \  Q(t) \ \le \ Q(\delta t) \ \le \ Q(t)^{\delta^2},
\quad \text{if \, $\delta \le 1$}
\end{alignat*}
\end{proposition}
\begin{proof}
Substituting the definition of $\lambda(t)$ in~\cref{lem:mills-bound}
and simplifying the expression, we obtain~(a).
Applying these bounds appropriately on both sides of~(b) proves that inequality as well.

Next, let  $(X,Y) \sim \bgs(\rho)$ so that  
$L(t,\eta;\rho) = \Pr\bigl[X > t, Y > \eta t]$. 
When $\rho=1$, we have $X=Y$ with probability 1 so $L(t,\eta;\rho) = Q(t)$,
implying~(c) trivially.
Therefore, let $\rho < 1$. 

Now $Y = \rho X + \sqrt{1-\rho^2} Z$ where 
$Z \sim \Normal(0,1)$ is independent of $(X,Y)$.
Observe:
\begin{align*}
\Pr[X > t, Y > \eta t] 
&= \Pr[X > t, \ \rho X + \sqrt{1-\rho^2} Z > \eta t] \\
&\ge \Pr[X > t, \ \rho t + \sqrt{1-\rho^2} Z > \eta t] \\
&= \Pr[X > t, Z > \delta t] && \text{(valid, because $\rho < 1$)}\\
&=  Q(t) \cdot Q(\delta t),
\end{align*}
proving~(c). 
For the last inequality, because $\delta \le 1$, we have $Q(t) \le Q(\delta t)$, 
and the latter can be bounded from above using
the first property in~\cref{lem:Q-tail-prop}, which implies~(d).
\end{proof}

\newpage
\section{Non-Asymptotic Bounds on Correlated Binomials} 
\label{sec:binom-apdx}

We let $h(\cdot)$ denote the binary entropy function.

\begin{fact}\label{fa:stirling_bin}
	Stirling's approximation of the factorial implies that for every integers $0 < \ell < m$, we have that
	\begin{equation*}
	\binom{m}{\ell} = \Theta\bigg( \sqrt{\frac{m}{\ell \cdot (m-\ell)}} \bigg) \cdot 2^{ - m \cdot h(\frac{\ell}{m})}.
	\end{equation*}
\end{fact}

\begin{fact}[Taylor approximation of binary entropy function]\label{fa:taylor_ent}
	For every $x \in [0,1]$, we have that
	\begin{equation*}
	h(1/2-x/2) = 1 - \frac{1}{2 \ln{2}} x^2 - O(x^4).
	\end{equation*}
\end{fact}

We now prove \Cref{prop:bin_tail_ub}.
\begin{proof}[Proof of \Cref{prop:bin_tail_ub}]
	We have that
	\begin{align*}
	\Pr_{X \in_R \{0,1\}^n}&[ \wt(X) \le n/2 - u \sqrt{n}/2] \\
	&= \displaystyle\sum\limits_{i=0}^{n/2 - u\sqrt{n}/2} \binom{n}{i} \cdot 2^{-n}\\ 
	&\overset{(A)}{\le} n \cdot 2^{-n} \cdot \binom{n}{n/2 - u\sqrt{n}/2}\\ 
	&= n \cdot 2^{-n} \cdot \Theta\bigg( \sqrt{\frac{n}{(n/2-u\sqrt{n}/2) \cdot (n/2+u\sqrt{n}/2)}} \bigg) \cdot 2^{n \cdot h\big( \frac{n/2-u\sqrt{n}/2}{n} \big)}\\ 
	&\overset{(B)}{\le} O(n) \cdot 2^{-n} \cdot 2^{n \cdot (1- \frac{u^2}{2 \cdot \ln{2} \cdot n})}\\ 
	&= O(n) \cdot e^{-\frac{u^2}{2}}\\ 
	&\overset{(C)}{\le} O(n^2) \cdot Q(u),
	\end{align*}
	where~$(A)$ follows from~\cref{fa:stirling_bin}, 
	~$(B)$ from~\cref{fa:taylor_ent}, 
	and~$(C)$ from~\cref{prop:QL-bound}\,(a). 
	Since the distribution of $\wt(X)$ is symmetric around $n/2$, the other case follows as well.
\end{proof}

We point out that in the statement of \Cref{prop:bin_tail_ub} we made no effort to optimize the multiplicative function of $n$ since that would not be consequential for our purposes. Recall that $r := n/2 - t\sqrt{n}/2$.
\begin{proposition}\label{prop:bin_tail_lb}
For any $k = o(\sqrt{n})$, we have that
\begin{equation*}
\Pr_{X \in_R \{0,1\}^n} [ \wt(X) \le r] \geq \frac{1}{\Theta(\sqrt{n})} \cdot Q(t).
\end{equation*}
\end{proposition}
\begin{proof}
We have that
\begin{align*}
\Pr_{X \in_R \{0,1\}^n} [ \wt(X) \le r] &= \displaystyle\sum\limits_{i=0}^r \binom{n}{i} \cdot 2^{-n}\\ 
&\geq 2^{-n} \cdot \binom{n}{r}\\ 
&= 2^{-n} \cdot \Theta\bigg(\sqrt{\frac{n}{r \cdot (n-r)}} \bigg) \cdot 2^{n \cdot h\big(\frac{n/2 - t\sqrt{n}/2}{n}\big)}\\ 
&\geq 2^{-n} \cdot \frac{1}{\Theta(\sqrt{n})} \cdot 2^{n \cdot h\big(\frac{n/2 - t\sqrt{n}/2}{n}\big)}\\ 
&= 2^{-n} \cdot \frac{1}{\Theta(\sqrt{n})} \cdot 2^{n \cdot \big( 1 - \frac{t^2}{2 \cdot \ln{2} \cdot n} - O\big(\frac{t^4}{n^2} \big)\big)}\\ 
&= \frac{1}{\Theta(\sqrt{n})} \cdot e^{-\frac{t^2}{2}}\\ 
&\geq \frac{1}{\Theta(\sqrt{n})} \cdot Q(t),
\end{align*}
where the second equality follows from \Cref{fa:stirling_bin}, the third equality follows from \Cref{fa:taylor_ent}, the fourth equality uses the assumption that $k = o(\sqrt{n})$ and the fact that $t = \Theta(\sqrt{k})$, and the last inequality follows from \Cref{prop:QL-bound}\,(a).
\end{proof}

\begin{lemma}\label{le:biased_tail}
Fix $\epsilon \in (0,0.5]$. For positive every $\alpha$ such that $\alpha^3 \cdot m = o_m(1)$, we have that
\begin{equation*}
\Pr[ \Bin(m, \epsilon) = (\epsilon+\alpha) \cdot m] \geq \Theta\bigg(\frac{1}{\sqrt{m}}\bigg) \cdot e^{-\frac{m \cdot \alpha^2}{2 \cdot \epsilon \cdot (1-\epsilon)}},
\end{equation*}
and simiarly,
\begin{equation*}
\Pr[ \Bin(m, \epsilon) = (\epsilon-\alpha) \cdot m] \geq \Theta\bigg(\frac{1}{\sqrt{m}}\bigg) \cdot e^{-\frac{m \cdot \alpha^2}{2 \cdot \epsilon \cdot (1-\epsilon)}}.
\end{equation*}
\end{lemma}

\begin{proof}
Stirling's approximation of the factorial implies that for every integers $0 < \ell < m$, we have that
\begin{equation}\label{eq:st_out}
\binom{m}{\ell} = \Theta\bigg(\sqrt{\frac{m}{\ell \cdot (m-\ell)}}\bigg) \cdot 2^{m \cdot h(\frac{\ell}{m})}.
\end{equation}
Applying \Cref{eq:st_out} with $\ell \triangleq (\epsilon+\alpha) \cdot m$, we get that
\begin{equation*}
\binom{m}{(\epsilon+\alpha) \cdot m} \geq \Theta\bigg(\frac{1}{\sqrt{m}}\bigg) \cdot 2^{m \cdot h(\epsilon + \alpha)}.
\end{equation*}
Thus,
\begin{align*}
\Pr[ \Bin(m, \epsilon) = (\epsilon+\alpha) \cdot m] 
&= \binom{m}{(\epsilon+\alpha) \cdot m} \cdot \epsilon^{(\epsilon+\alpha) \cdot m} \cdot (1-\epsilon)^{m-(\epsilon+\alpha) \cdot m}\\ 
&\geq \Theta\bigg(\frac{1}{\sqrt{m}}\bigg) \cdot 2^{m \cdot h(\epsilon + \alpha)} \cdot 2^{m \cdot ( - h(\epsilon) + \alpha \cdot \log(\frac{\epsilon}{1-\epsilon}))}\\ 
&= \Theta\bigg(\frac{1}{\sqrt{m}}\bigg) \cdot 2^{m \cdot \big( h(\epsilon + \alpha) -h(\epsilon) + \alpha \cdot \log(\frac{\epsilon}{1-\epsilon})\big)}.
\end{align*}
Note that that for every $\epsilon > 0$,
\begin{equation*}
h'(\epsilon) = - \log\big( \frac{\epsilon}{1-\epsilon} \big),
\end{equation*}
and
\begin{equation*}
h''(\epsilon) = -\frac{1}{\ln{2} \cdot \epsilon \cdot (1-\epsilon)}.
\end{equation*}
Taylor expanding $h(\epsilon +\alpha)$ around $\epsilon >0$, we get that
\begin{align*}
h(\epsilon+\alpha) &= h(\epsilon) + h'(\epsilon) \cdot \alpha + \frac{h''(\epsilon) \cdot {\alpha^2}}{2} \pm O_{\epsilon}(\alpha^3)\\ 
&= h(\epsilon) - \alpha \cdot \log\big( \frac{\epsilon}{1-\epsilon} \big) -\frac{\alpha^2}{2 \cdot \ln{2} \cdot \epsilon \cdot (1-\epsilon)} \pm O_{\epsilon}(\alpha^3).
\end{align*}
Thus, we get that
\begin{align*}
\Pr[ \Bin(m, \epsilon) = (\epsilon+\alpha) \cdot m] &\geq \Theta\bigg(\frac{1}{\sqrt{m}}\bigg) \cdot 2^{m \cdot \big(-\frac{\alpha^2}{2 \cdot \ln{2} \cdot \epsilon \cdot (1-\epsilon)} \pm O_{\epsilon}(\alpha^3)\big)}\\ 
&= \Theta\bigg(\frac{1}{\sqrt{m}}\bigg) \cdot e^{-\frac{m \cdot \alpha^2}{2 \cdot \epsilon \cdot (1-\epsilon)}},
\end{align*}
where the last equality uses the given assumption that $\alpha^3 \cdot m = o_m(1)$. The proof of the second part of the lemma follows along the same line with $\alpha$ being replaced by $-\alpha$.
\end{proof}

\begin{proposition}\label{prop:first_cond_w}
Fix $\epsilon \in (0,0.5]$. For every $n = \omega (k^3)$, we have that
\begin{equation*}
\Pr_{(X,Y) \sim \dsbs(1-2\epsilon)^{\otimes n}}[ Y \in \Ball(0, r') | X \in \Ball(0,r)] \geq \Theta\bigg(\frac{1}{n^{1.5}}\bigg) \cdot 2^{-\delta^2 k},
\end{equation*}
where $\Ball(0, r)$ denotes the Hamming ball of radius $r$ centered around the all-zeros vector.
\end{proposition}

\begin{proof}
We start by showing that if $A \sim \Bin(\epsilon, n/2 + t\sqrt{n}/2)$ and $B \sim \Bin(\epsilon, n/2 - t\sqrt{n}/2)$ are independent random variables, then
\begin{equation}\label{eq:bas_ineq}
\Pr[A \le B + r'-r] \geq \Theta\bigg(\frac{1}{n}\bigg) \cdot e^{-\frac{t^2 \delta^2}{2}}.
\end{equation}
To prove \Cref{eq:bas_ineq}, note that
\begin{align*}
\Pr[A \le B + r'-r] &\geq \Pr[A = \frac{\epsilon \cdot n +0.5 \eta t \sqrt{n}}{2}, B = \frac{\epsilon \cdot n - 0.5 \eta t \sqrt{n}}{2}]\\ 
&= \Pr[A = \frac{\epsilon \cdot n + 0.5 \eta t \sqrt{n}}{2}] \cdot \Pr[B = \frac{\epsilon \cdot n - 0.5 \eta t \sqrt{n}}{2}]
\end{align*}
Applying \Cref{le:biased_tail} with $m = n/2 + t\sqrt{n}/2$ and $\alpha = \frac{-\epsilon t +0.5 \eta t}{\sqrt{n} +t }$, we get that
\begin{align*}
\Pr[A = \frac{\epsilon \cdot n + 0.5 \eta t \sqrt{n}}{2}]  &\geq \Theta\bigg(\frac{1}{\sqrt{n}}\bigg) \cdot e^{-\frac{(\epsilon - 0.5\eta)^2 \cdot t^2}{4 \cdot \epsilon \cdot (1-\epsilon)}}.
\end{align*}
Similarly, applying \Cref{le:biased_tail} with $m = n/2 -t\sqrt{n}/2$ and $\alpha = \frac{\epsilon t - 0.5 \eta t}{\sqrt{n} - t}$, we get that
\begin{align*}
\Pr[B = \frac{\epsilon \cdot n - 0.5 \eta t \sqrt{n}}{2}] &\geq \Theta\bigg(\frac{1}{\sqrt{n}}\bigg) \cdot e^{-\frac{(\epsilon -0.5 \eta)^2 \cdot t^2}{4 \cdot \epsilon \cdot (1-\epsilon)}}. 
\end{align*}
Note that when applying \Cref{le:biased_tail}, we have used the assumption that $n = \omega(k^3)$ and the fact that $t = \Theta(\sqrt{k})$. Thus, we get that
\begin{align*}
\Pr[A \le B +r' -r] &\geq \Theta\bigg(\frac{1}{n}\bigg) \cdot e^{-\frac{t^2}{4\cdot \epsilon \cdot (1-\epsilon)} \cdot \big( (\epsilon - 0.5\eta)^2 + (\epsilon -0.5 \eta)^2\big)}\\ 
&= \Theta\bigg(\frac{1}{n}\bigg) \cdot e^{-\frac{t^2}{8\cdot \epsilon \cdot (1-\epsilon)} \cdot (2\epsilon - \eta)^2}\\ 
&= \Theta\bigg(\frac{1}{n}\bigg) \cdot e^{-\frac{t^2 \delta^2}{2}},
\end{align*}
where the last equality follows from the fact that $n = \omega(t^4)$, which in particular follows from the assumption that $n = \Omega(k^3)$ and the fact that $t = \Theta(\sqrt{k})$. \Cref{eq:bas_ineq} now implies that
\begin{equation*}
\Pr_{(X,Y) \sim \dsbs(1-2\epsilon)^{\otimes n}}[ Y \in \Ball(0, r') | \wt(X) = r] \geq \Theta\bigg(\frac{1}{n}\bigg) \cdot 2^{- \delta^2 k},
\end{equation*}
where $\wt(X)$ denotes the Hamming weight of $X$. The statement of \Cref{prop:first_cond_w} now follows from the fact that
$$ \Pr_{X \in_R \{0,1\}^n}[ \wt(X) = r  | X \in \Ball(0,r)] \geq \Theta\bigg(\frac{1}{\sqrt{n}} \bigg),$$
which itself uses the fact that $r \le n/2$.
\end{proof}

In order to prove \Cref{lem:sec_cond_w}, we will need the following propositions.

\begin{proposition}\label{prop:monot}
Let $t_2 \geq 0$ and $a_{max} \triangleq n \cdot (1+\theta)/4 - (t+t_2)\cdot \sqrt{n}/4$. For every $a \in \{0,1,\dots, a_{max} \}$, let
\begin{equation*}
\psi(a) \triangleq \frac{ \binom{n \cdot (1+\theta)/2}{a} \cdot \binom{n \cdot (1-\theta)/2}{n/2 - t_2\sqrt{n}/2 - a}}{\binom{n}{n/2 - t_2\sqrt{n}/2}}.
\end{equation*}
Then, $\psi(a)$ is monotonically increasing in $a$.
\end{proposition}
\begin{proof}
Let $a \in \{1,\dots, a_{max}\}$. Then,
\begin{align*}
\frac{\psi(a)}{\psi(a-1)} &= \frac{(n/2 \cdot(1+\theta) +1 -a) \cdot (n/2 -t_2 \cdot \sqrt{n}/2 + 1 -a)}{a \cdot (t_2 \cdot \sqrt{n}/2 - \theta \cdot n/2 +a)}.
\end{align*}
This implies that $\psi(a) \geq \psi(a-1)$ if and only if
\begin{equation*}
a \le \frac{(n/2 \cdot (1+\theta) +1) \cdot (n/2 - t_2 \cdot \sqrt{n}/2 +1)}{n+2},
\end{equation*}
which is satisfied by all $a \in \{0,1,\dots, a_{max} \}$ (for large enough $n$).
\end{proof}

\begin{proposition}\label{prop:psi_a_max}
Assume that $t = o(n^{1/4})$, $t_2 = o(n^{1/4})$ and $\theta \cdot t \cdot t_2 = o_n(1)$. Then,
$$\psi(a_{max}) \le \Theta\bigg(\frac{1}{\sqrt{n}}\bigg) \cdot e^{- \frac{t^2}{2}}.$$
\end{proposition}
\begin{proof}
By \Cref{fa:stirling_bin}, we have that
\begin{align}\label{al:num_1}
\binom{n \cdot (1+\theta)/2}{a_{max}} &= \Theta\bigg( \sqrt{\frac{n \cdot (1+\theta)/2}{a_{max} \cdot (n \cdot (1+\theta)/2-a_{max})}} \bigg) \cdot 2^{ - n \cdot \frac{(1+\theta)}{2} \cdot h\big(\frac{a_{max}}{n \cdot (1+\theta)/2}\big)}\nonumber\\ 
&= \Theta\bigg(\frac{1}{\sqrt{n}}\bigg) \cdot 2^{ - n \cdot \frac{(1+\theta)}{2} \cdot h\big(\frac{1}{2} - \frac{(t+t_2)}{2 \cdot(1+\theta) \cdot \sqrt{n}}\big)}
\end{align}
(where the second equality uses the assumptions that $t = o(\sqrt{n})$, $t_2 = o(\sqrt{n})$ and $\theta = o_n(1)$),

\begin{align*}
&\binom{n \cdot (1-\theta)/2}{n/2 - t_2\sqrt{n}/2 - a_{max}} \\
&\hphantom{123456789} 
=\Theta\bigg( \sqrt{\frac{n (1-\theta)/2}{(\frac{n}{2} - \frac{t_2\sqrt{n}}{2} - a_{max}) (\frac{t_2 \sqrt{n}}{2} - \frac{n \theta}{2} + a_{max})}} \bigg) 2^{ - n \frac{(1-\theta)}{2} h\big(\frac{n/2 - t_2\sqrt{n}/2 - a_{max}}{n \cdot (1-\theta)/2}\big)} \\ 
&\hphantom{123456789} 
=\Theta\bigg(\frac{1}{\sqrt{n}}\bigg) \cdot 2^{ - n \cdot \frac{(1-\theta)}{2} \cdot h\big( \frac{1}{2} - \frac{(t_2-t)}{2\cdot (1-\theta) \cdot \sqrt{n}}\big)}
\numberthis \label{al:num_2}
\end{align*}
and
\begin{align*}
\binom{n}{n/2 - t_2\sqrt{n}/2} &= \Theta\bigg( \sqrt{\frac{n}{(n/2 - t_2\sqrt{n}/2) \cdot (n/2 + t_2\sqrt{n}/2)}} \bigg) \cdot 2^{ - n \cdot h\big(\frac{n/2 - t_2\sqrt{n}/2}{n}\big)}\nonumber\\ 
&= \Theta\bigg(\frac{1}{\sqrt{n}}\bigg) \cdot 2^{ - n \cdot h\big( \frac{1}{2} - \frac{t_2}{2 \cdot \sqrt{n}}\big)},
\numberthis \label{al:denom}
\end{align*}
where the second equality uses the assumption that $t_2 = o(\sqrt{n})$. Combining \Cref{al:num_1}, \Cref{al:num_2} and \Cref{al:denom}, we get that
\begin{align}\label{al:psi_ent}
\psi(a_{max}) &= \Theta\bigg(\frac{1}{\sqrt{n}}\bigg) \cdot 2^{-n \cdot \bigg( \frac{(1+\theta)}{2} \cdot h\big(\frac{1}{2} - \frac{(t+t_2)}{2 \cdot(1+\theta) \cdot \sqrt{n}}\big) + \frac{(1-\theta)}{2} \cdot h\big( \frac{1}{2} - \frac{(t_2-t)}{2\cdot (1-\theta) \cdot \sqrt{n}}\big) - h\big( \frac{1}{2} - \frac{t_2}{2 \cdot \sqrt{n}}\big)  \bigg)}.
\end{align}
By \Cref{fa:taylor_ent}, we have that
\begin{align}\label{al:expo}
&\frac{(1+\theta)}{2} \cdot h\big(\frac{1}{2} - \frac{(t+t_2)}{2 \cdot(1+\theta) \cdot \sqrt{n}}\big) + \frac{(1-\theta)}{2} \cdot h\big( \frac{1}{2} - \frac{(t_2-t)}{2\cdot (1-\theta) \cdot \sqrt{n}}\big) - h\big( \frac{1}{2} - \frac{t_2}{2 \cdot \sqrt{n}}\big)\nonumber\\ 
&= \frac{(1+\theta)}{2} \big(1 - \frac{(t+t_2)^2}{2 \ln{2} \cdot (1+\theta)^2 n}\big) +  \frac{(1-\theta)}{2} \big(1 - \frac{(t-t_2)^2}{2 \ln{2} \cdot (1-\theta)^2 n}\big) - \big(1- \frac{t_2^2}{2 \ln{2} \cdot n}\big) \pm O\big( \frac{(t+t_2)^4}{n^2} \big)\nonumber\\ 
&= \frac{t^2}{ 2 \ln{2} \cdot n} + \frac{(\theta^2 \cdot t^2 -4 \cdot \theta \cdot t \cdot t_2 + 2 \cdot \theta^2 \cdot t_2^2)}{4 \ln{2} \cdot n} \pm O\big( \frac{(t+t_2)^4}{n^2} \big),
\end{align}
where the second equality above uses the fact that $\theta \le 1$. Plugging \Cref{al:expo} back in \Cref{al:psi_ent}, we get that
\begin{align*}
\psi(a_{max}) &= \Theta\bigg(\frac{1}{\sqrt{n}}\bigg) \cdot 2^{- \frac{t^2}{ 2 \ln{2}} - \frac{(\theta^2 \cdot t^2 -4 \cdot \theta \cdot t \cdot t_2 + 2 \cdot \theta^2 \cdot t_2^2)}{4 \ln{2}} \pm O\big( \frac{(t+t_2)^4}{n} \big)}\\ 
&\le \Theta\bigg(\frac{1}{\sqrt{n}}\bigg) \cdot 2^{- \frac{t^2}{ 2 \ln{2}}}\\ 
&= \Theta\bigg(\frac{1}{\sqrt{n}}\bigg) \cdot e^{- \frac{t^2}{2}},
\end{align*}
where the inequality uses the assumptions that $t = o(n^{1/4})$, $t_2 = o(n^{1/4})$ and $\theta \cdot t \cdot t_2 = o_n(1)$.
\end{proof}

\newpage
\section{Correlated Randomness Generation}\label{sec:gcr-append}

We first recall that Canonne et al. \cite{CGMS_ISR} -- using the converse bound of \cite{bogdanov2011extracting} -- showed that for any $\epsilon > 0$, if Alice and Bob are given access to i.i.d. samples from $\dsbs(1-2\epsilon)$, then, \emph{perfectly agreeing} on $k$ random bits requires $\Omega_{\epsilon}(k)$ bits of communication even in the two-way model. They also raised the following intriguing question: ``What if their goal is only to generate more correlated bits than they start with? What is possible here and what are the limits?''

We partially answer this question and show that for any $\epsilon > 0$ and $\epsilon' \gg \epsilon \cdot \log(1/\epsilon)$, if Alice and Bob are given access to i.i.d. samples from $\dsbs(1-2\epsilon')$, then, generating $k$ random samples from $\dsbs(1-2\epsilon)$ requires $\Omega_{\epsilon,\epsilon'}(k)$ bits of communication.

\begin{definition}[Correlated Randomness Generation]
In the CorrelatedRandomness$_{\gamma,\epsilon',\alpha,k}$ problem, Alice and Bob are given access to i.i.d. samples from a known source/distribution $\mu$. 
Their goal is to for Alice to output $w_A \in \{0,1\}^k$ and for Bob to output $w_B \in \{0,1\}^k$, that satisfy the following properties: 
(i) $\Pr[\Delta(w_A,w_B) \le \epsilon' k] \geq \gamma$;
(ii) $H_{\infty}(w_A) \geq \alpha \cdot k$; and 
(iii) $H_{\infty}(w_B) \geq \alpha \cdot k$.

\end{definition}

We point out that one can alternatively define Correlated Randomness Generation in terms of coming close, say  in total variation distance, to $\dsbs(1-2\eps')^{\otimes k}$. The results in this section apply to this variant as well. This is because of the next lemma which can be proved by a simple Chernoff bound and which says that if Alice and Bob are given access to i.i.d. samples from $\dsbs(1-2\epsilon')$, then they can generate two length-$k$ binary strings that lie in a Hamming ball of radius $\approx \epsilon' \cdot k$ with high probability.

\begin{lemma}\label{le:ham_ball_chernoff}
Fix $0 < \delta < \epsilon'$ and let $\dsbs(1-2(\epsilon'-\delta))$
be the source. 
Then, there is a non-interactive protocol solving CorrelatedRandomness$_{\gamma,\epsilon',\alpha,k}$ with $\gamma = 1-\exp(-(\epsilon'-\delta)^2 \cdot k)$ and $\alpha = 1$.
\end{lemma}

We are now ready to state the main result.
\begin{theorem}[Interactive Correlated Randomness Generation]\label{th:gcr_formal}
Any interactive protocol solving\\  CorrelatedRandomness$_{\gamma,\epsilon',\alpha,k}$ 
for the source $\dsbs(1-2(\epsilon'-\delta))$
with $h(\epsilon') \le 4 \cdot \epsilon \cdot (1-\epsilon) \cdot \alpha /(1+\Omega(1))$ should communicate at least
$\Omega( \epsilon \cdot \alpha \cdot k) - O(\log(1/\gamma))$ bits.
\end{theorem}

\Cref{le:gen_ag_dist_imposs} says that \emph{non-interactively} generating two strings with min-entropy $k$ and that lie in a Hamming ball of radius $\approx \epsilon' \cdot k$ cannot be done with success probability $2^{-o_{\epsilon}(k)}$ when Alice and Bob are given access to i.i.d. samples from $\dsbs(1-2\epsilon)$ with $\epsilon = \omega( \epsilon' \cdot \log(1/\epsilon'))$.

\begin{theorem}[Non-Interactive Correlated Randomness Generation]\label{le:gen_ag_dist_imposs}
There is no non-interactive protocol solving CorrelatedRandomness$_{\gamma,\epsilon',\alpha,k}$ for the source $\dsbs(1-2\epsilon)$ with $h(\epsilon') \le 4 \cdot \epsilon \cdot (1-\epsilon) \cdot \alpha$ and 
$\gamma > 2^{-\nu k}$ where 
\begin{equation*}
\nu = \alpha \cdot \frac{[\sqrt{1 - h(\epsilon')/\alpha} - (1-2\epsilon)]^2}{4 \cdot \epsilon \cdot (1-\epsilon)}.
\end{equation*}
Consequently, whenever $h(\epsilon') \le 4 \cdot \epsilon \cdot (1-\epsilon) \cdot \alpha /(1+\Omega(1))$, there is no non-interactive protocol solving CorrelatedRandomness$_{\gamma,\epsilon',\alpha,k}$ given i.i.d. access to $\dsbs(1-2\epsilon)$ with $\gamma > 2^{-\Omega(\epsilon \cdot \alpha \cdot k)}$.
\end{theorem}

We point out that getting the tight bounds in \Cref{th:gcr_formal} and \Cref{le:gen_ag_dist_imposs} remains a very interesting open question. In order to prove \Cref{th:gcr_formal} and \Cref{le:gen_ag_dist_imposs}, we next introduce a ``list'' version of Common Randomness which is implicit in several of the known converse results for Common Randomness Generation.

\begin{definition}[List Common Randomness Generation]
In the ListCommonRandomness$_{\gamma, b}^k$ problem, Alice and Bob are given access to i.i.d. samples from a known distribution $\mu$ over pairs of random variables. Their goal is for Alice to output an element $w_A$ and for Bob to output a list $L_B$ (over the same universe), such that
(i) $\Pr[w_A \in L_B] \geq \gamma$;
(ii) $H_{\min}(w_A) \geq k$; and
(iii) $|L_B| \le b$.

\end{definition}

We prove the following converse results for List Common Randomness Generation both in the non-interactive and two-way communication models:

\begin{theorem}[Non-Interactive List Common Randomness Generation]\label{th:non_int_list_cr}
There is no non-interactive protocol solving ListCommonRandomness$_{\gamma, b}^k$ for the source $\dsbs(1-2\epsilon)$ with $(\log{b})/k \le 4 \cdot \epsilon \cdot (1-\epsilon)$ and with $\gamma > 2^{-\nu k}$ where
\begin{equation*}
\nu = \frac{[\sqrt{1 - (\log{b})/k}-(1-2\epsilon)]^2}{4 \cdot \epsilon \cdot (1-\epsilon)}.
\end{equation*}
Consequently, whenever $(\log{b})/k \le 4 \cdot \epsilon \cdot (1-\epsilon)/(1+\Omega(1))$, there is no non-interactive protocol solving ListCommonRandomness$_{\gamma, b}^k$ with $\gamma > 2^{-\Omega(\epsilon \cdot k)}$.
\end{theorem}
\begin{proof}
The proof is very similar to that of the converse result of \cite{guruswami2016tight}. Let $\Pi$ be a protocol solving ListCommonRandomness$_{\gamma, b}^k$. Let $X$ be Alice's input and $w_A \triangleq f(X)$ be her output, and let $Y$ be Bob's input and $L_B \triangleq (g_1(Y), g_2(Y), \dots, g_b(Y))$ be his output. Here, $(X,Y) \sim \dsbs(1-2\epsilon)^{\otimes n}$, and $f$, $g_1$, $g_2$, $\dots$, $g_b$ are functions mapping $\{0,1\}^n$ to $\{0,1\}^k$. For every $y \in \{0,1\}^n$ and $z \in \{0,1\}^k$, denote $\beta(z | y) \triangleq \Pr[f(X) = z | Y =y]$. The success probability of the protocol $\Pi$ is given by
\begin{align*}
\Pr[w_A \in L_B] &= \Pr[f(X) \in \{g_1(Y), g_2(Y), \dots, g_b(Y)\}]\\ 
&= \Ex_{y}[ \Pr[ f(X) \in L_B(y) ~ | ~ Y= y]\\ 
&= \Ex_{y}[ \displaystyle\sum\limits_{z \in L_B(y)} \beta(z | y)]\\ 
&\le \Ex_{y}[ (\displaystyle\sum\limits_{z \in L_B(y)} \beta(z | y)^q)^{1/q}] \cdot b^{1-1/q}\\ 
&\le \Ex_{y}[ (\displaystyle\sum\limits_{z} \beta(z | y)^q)^{1/q}] \cdot b^{1-1/q}\\ 
&\le (\Ex_y [\displaystyle\sum\limits_{z} \beta(z | y)^q])^{1/q} \cdot b^{1-1/q}\\ 
&= (\displaystyle\sum\limits_{z} \Ex_y[\beta(z | y)^q])^{1/q} \cdot b^{1-1/q},
\end{align*}
where the first inequality follows from Holder's inequality and the last inequality follows from the fact that the function $x \mapsto x^{1/q}$ for non-negative $x$ is concave for every $q \geq 1$. Consider the function $h_z \colon \{0,1\}^n \to \{0,1\}$ given by $h_z(X) = \mathbbm{1}[f(X) = z]$ for all $X \in \{0,1\}^n$. Hypercontractivity then implies that
\begin{align*}
\Ex_y[\beta(z | y)^q])^{1/q} &= \Ex_y [\Ex[ h_z(X) ~ | ~ Y=y]^q]\\ 
&= \| \Ex[h_z(X) ~ |~ Y] \|_q^q\\ 
&\le \|h_z \|_p^q\\ 
&= (\Ex_x h_z(x))^{q/p}\\ 
&= \Pr[f(X) = z]^{q/p}.
\end{align*}
Thus, the success probability of $\Pi$ satisfies
\begin{align*}
\Pr[w_A \in L_B] &\le (\displaystyle\sum\limits_{z} \Pr[f(X) = z]^{q/p})^{1/q} \cdot b^{1-1/q}\\ 
&= (\displaystyle\sum\limits_{z} \Pr[f(X) = z]^{q/p-1} \cdot \Pr[f(X) = z])^{1/q} \cdot b^{1-1/q}\\ 
&\le ( 2^{-k \cdot (\frac{q}{p}-1)} \cdot \displaystyle\sum\limits_{z} \Pr[f(X) = z])^{1/q} \cdot b^{1-1/q}\\ 
&=2^{-k \cdot (q/p-1) \cdot \frac{1}{q}} \cdot b^{1- 1/q},
\end{align*}
where the inequality above follows from the fact that $w_A$ has min-entropy at least $k$ bits. Setting $p = 1+(1-2 \cdot \epsilon)^2 \cdot \delta$ and $q = 1+\delta$ and optimizing for $\delta$, we get that
\begin{equation*}
\gamma \le 2^{-k \cdot \frac{[-\sqrt{s} + \sqrt{1 - (\log{b})/k}]^2}{1-s}},
\end{equation*}
where $s = (1-2\epsilon)^2$ is the Strong Data Processing Constant of the $\dsbs(1-2\epsilon)$ source, and where the above bound holds assuming that $(\log{b})/k \le 1-s$. The theorem statement now follows.
\end{proof}

We point out that \Cref{th:non_int_list_cr}  implies a lower bound on the $1$-way communication complexity of List Common Randomness Generation (by essentially increasing the list size by a factor of $2^c$ where $c$ is the communication from Alice to Bob). It turns out that, by adapting a reduction of \cite{CGMS_ISR}, one can also use \Cref{th:non_int_list_cr} to get a lower bound on the \emph{interactive} communication complexity of List Common Randomness Generation, which we state next.

\begin{theorem}[Interactive List Common Randomness Generation]\label{th:int_list_cr}
Let $\dsbs(1-2\epsilon)$ be the source.
Then, any interactive protocol solving ListCommonRandomness$_{\gamma, b}^k$ with $(\log{b})/k \le 4 \cdot \epsilon \cdot (1-\epsilon)$ should communicate at least
$$k \cdot \frac{[\sqrt{1 - (\log{b})/k}-(1-2\epsilon)]^2}{8\cdot \epsilon \cdot (1-\epsilon)} - \frac{3}{2}\log(1/\gamma) - O(1) ~~~ \text{ bits.}$$
Consequently, whenever $(\log{b})/k \le 4 \cdot \epsilon \cdot (1-\epsilon)/(1+\Omega(1))$, any interactive protocol solving \\ ListCommonRandomness$_{\gamma, b}^k$ should communicate at least $\Omega(\epsilon \cdot k) - O(\log{1/\gamma})$ bits.
\end{theorem}
\begin{proof}
The proof will combine \Cref{th:non_int_list_cr} with the approach of \cite{CGMS_ISR} for getting lower bounds on \emph{interactive} Common Randomness Generation using lower bounds on \emph{non-interactive} Common Randomness Generation. 

Let $\Pi$ be an interactive protocol solving ListCommonRandomness$_{\gamma, b}^k$ with $(\log{b})/k \le (1-s)/(1+\Omega(1))$. Let $X$ denote Alice's input and $Y$ denote Bob's input. Consider now the non-interactive protocol $\Pi$ where on input pair $(X,Y)$:
\begin{enumerate}
\item 
Alice samples $Y'$ from the conditional distribution of $\mu$ 
given $X$, and she outputs the element that she would have output in the execution of $\Pi$ on $(X,Y')$.
\item Bob samples $X'$ from the conditional distribution of $\mu$ given $Y$, and he outputs the list that he would have output in the execution of $\Pi$ on $(X',Y)$.
\end{enumerate}
Note that the non-inteactive protocol $\Pi'$ satisfies the property that the min-entropy of Alice's output is at least $k$ (since it is exactly equal to the min-entropy of Alice's output under $\Pi$). We next show that the success probability of the protocol $\Pi'$ is at least $\Omega(\gamma^3 \cdot 2^{- 2 \cdot c})$ where $c$ is the two-way communication complexity of $\Pi$. Using \Cref{th:non_int_list_cr}, this would imply that
$$c \geq k \cdot \frac{[\sqrt{1 - (\log{b})/k}-(1-2\epsilon)]^2}{8\cdot \epsilon \cdot (1-\epsilon)} - \frac{3}{2}\log(1/\gamma) - O(1),$$
which implies the desired statement. We now lower-bound the success probability of $\Pi'$. Let $P_X(t)$ denote the probability over $Y'$ conditioned on $X$ that $\Pi(X,Y')$ is equal to the transcript $t$. Similarly, let $Q_Y(t)$ denote the probability over $X'$ conditioned on $Y$ that $\Pi(X',Y)$ is equal to the transcript $t$. Let $G$ be the set of all input pairs $(X,Y)$ such that, in the execution of $\Pi(X,Y)$, Alice's output element belongs to Bob's output list. Then, the success probability of $\Pi$ is equal to
\begin{equation*}
\gamma = \displaystyle\sum\limits_{(X,Y) \in G} \mu(X,Y).
\end{equation*}
We say that a transcript $t$ is unlikely for $X$ if $P_X(t) < (\gamma/4) \cdot 2^{-c}$. Similarly, we say that a transcript $t$ is unlikely for $Y$ if $Q_Y(t) < (\gamma/4) \cdot 2^{-c}$. Let $B$ be the set of all input-pairs $(X,Y)$ such that the transcript $\Pi(X,Y)$ is either unlikely for $X$ or unlikely for $Y$. Note that
\begin{align}\label{eq:unlik_X}
\displaystyle\sum\limits_{(X,Y): ~ \Pi(X,Y) \text{ unlikely for } X} \mu(X,Y) ~~ &= ~~  \displaystyle\sum\limits_{X} ~~ \displaystyle\sum\limits_{t \text{ unlikely for } X}  ~~ \displaystyle\sum\limits_{Y: ~ \Pi(X,Y) = t} \mu(X,Y)\nonumber\\ 
~~ &= ~~ \displaystyle\sum\limits_{X} \mu(X) \cdot \displaystyle\sum\limits_{t \text{ unlikely for } X} P_X(t)\nonumber\\ 
~~ &< ~~ \displaystyle\sum\limits_{X} \mu(X) \cdot \displaystyle\sum\limits_{t \text{ unlikely for } X} \frac{\gamma}{4} \cdot 2^{-c}\nonumber\\ 
~~ &< ~~ \frac{\gamma}{4}.
\end{align}
An identical argument shows that
\begin{equation}\label{eq:unlik_Y}
\displaystyle\sum\limits_{(X,Y): ~ \Pi(X,Y) \text{ unlikely for } Y} \mu(X,Y) < \frac{\gamma}{4}.
\end{equation}
Combining \Cref{eq:unlik_X} and \Cref{eq:unlik_Y}, we get that
\begin{equation*}
\displaystyle\sum\limits_{(X,Y) \in B} \mu(X,Y) < \frac{\gamma}{2}.
\end{equation*}

The success probability of $\Pi'$ can now be lower-bounded by
\begin{align*}
\displaystyle\sum\limits_{(X,Y) \in G} \mu(X,Y) \cdot P_X(\Pi(X,Y)) \cdot Q_Y( \Pi(X,Y)) &\geq \displaystyle\sum\limits_{(X,Y) \in G \setminus B} \mu(X,Y) \cdot P_X(\Pi(X,Y)) \cdot Q_Y( \Pi(X,Y))\\ 
&\geq \displaystyle\sum\limits_{(X,Y) \in G \setminus B} \mu(X,Y) \cdot \frac{\gamma^2}{16} \cdot 2^{- 2 \cdot c}\\ 
&= \frac{\gamma^2}{16} \cdot 2^{- 2 \cdot c} \cdot \bigg( \displaystyle\sum\limits_{(X,Y) \in G} \mu(X,Y) - \displaystyle\sum\limits_{(X,Y) \in B} \mu(X,Y) \bigg)\\ 
&\geq \frac{\gamma^3}{32} \cdot 2^{-2 \cdot c},
\end{align*}
as desired.
\end{proof}

We note that \Cref{th:non_int_list_cr} and \Cref{th:int_list_cr} also hold with the same bounds when the source is $\bgs(1-2\eps)$ instead of $\dsbs(1-2\eps)$. We now show how \Cref{th:non_int_list_cr} implies \Cref{le:gen_ag_dist_imposs}, and how \Cref{th:int_list_cr} implies \Cref{th:gcr_formal}.

\begin{proof}[Proof of \Cref{le:gen_ag_dist_imposs}]
Given a protocol $\Pi$ for CorrelatedRandomness$_{\gamma,\epsilon',\alpha,k}$, we give a protocol $\Pi'$ for ListCommonRandomness$_{\gamma, b}^{\alpha \cdot k}$ with $b \le 2^{h(\epsilon') \cdot k}$ as follows:
\begin{enumerate}
\item If $w_A$ is the output of Alice under the protocol $\Pi$, then she also outputs $w_A$ under the protocol $\Pi'$.
\item If $w_B$ is the output of Bob under the protocol $\Pi$, then he outputs the list $L_B \triangleq \Ball(w_B, \epsilon' \cdot k)$ under the protocol $\Pi'$.
\end{enumerate}
\Cref{le:gen_ag_dist_imposs} now follows from \Cref{th:non_int_list_cr} and the fact that $|\Ball(w_a, \epsilon' \cdot k)| \le 2^{h(\epsilon') \cdot k}$.
\end{proof}

\begin{proof}[Proof of \Cref{th:gcr_formal}]
The proof is identical to that of \Cref{le:gen_ag_dist_imposs} except that we use \Cref{th:int_list_cr} instead of \Cref{th:non_int_list_cr}.
\end{proof}

\newpage
\section{Communication with Imperfect Shared Randomness}\label{sec:smp-append}

We start by stating the most general result for this problem that applies
to many sources of randomness including $\dsbs(\rho)$.

\begin{theorem}\label{th:smp_rest}
Let $\rho \in (0,1]$ and $\mu$ be any source of randomness with maximal correlation $\rho$. Every (possibly partial) function $f$ with $(1/3)$-error two-way communication $c$ bits with perfect randomness has $\delta$-error zero-communication communication with $\mu$-randomness at most $2^{O(c)} \cdot \log(1/\delta) /\rho^2$ bits for every $\delta >0$.
\end{theorem}

We point out that the above theorem yields $\dsbs(\rho)$ as a special case
because of the fact (due to \cite{witsenhausen1975sequences}) 
that the maximal correlation of $\dsbs(\rho)$ is equal to $\rho$.

In order to prove \Cref{th:smp_rest}, we will give a zero-communication protocol with $\mu$-randomness (where $\mu$ is any source of randomness with maximal correlation $\rho$) solving the following  problem which is equivalent to ``sketching $\ell_2$-norms on the unit sphere.'' This problem was studied by \cite{CGMS_ISR} to prove a $1$-way (instead of a zero-communication) analogue of \Cref{th:smp_rest}.

\begin{definition}[GapInnerProduct$_{r,s}^n$]\label{def:GIP}
Let $-1 \le s < r \le 1$ be known to Alice and Bob. Alice is also given a unit vector $u \in \mathbb{R}^n$ and Bob is given a unit vector $v$ in $\mathbb{R}^n$. The goal is for Alice and Bob to distinguish the case where $\ip{u}{v} \geq r$ from the case where $\ip{u}{v} \le s$.
\end{definition}

The next lemma shows that GapInnerProduct is complete for functions with low interactive communication complexity.
\begin{lemma}[\cite{CGMS_ISR}]\label{le:cgms_comp}
Let $f$ be a (possibly partial) two-party function $f \colon \{0,1\}^{2 \cdot n} \to \{0,1\}$, such that $f$ has $(1/3)$-error two-way communication complexity $c$ bits with perfect randomness. Then, there exists a function $\ell(n) \in \mathbb{N}$ along with mappings $g_A \colon \{0,1\}^n \to \{\pm \frac{1}{\sqrt{\ell(n)}}\}^{\ell(n)}$ and $g_B \colon \{0,1\}^n \to \{\pm \frac{1}{\sqrt{\ell(n)}}\}^{\ell(n)}$ such that
\begin{enumerate}
\item If $f(x,y) = 0$, then $(g_A(x), g_B(y))$ is a NO instance of GapInnerProduct$_{\frac{2}{3} \cdot 2^{-k}-1, \frac{1}{3} \cdot 2^{-k}-1}^{\ell(n)}$. Namely, $\ip{g_A(x)}{g_B(y)} \le \frac{1}{3} \cdot 2^{-k}-1$.
\item If $f(x,y) = 1$, then $(g_A(x), g_B(y))$ is a YES instance of GapInnerProduct$_{\frac{2}{3} \cdot 2^{-k}-1, \frac{1}{3} \cdot 2^{-k}-1}^{\ell(n)}$. Namely, $\ip{g_A(x)}{g_B(y)} \geq \frac{2}{3} \cdot 2^{-k}-1$.
\end{enumerate}
\end{lemma}

The following theorem gives a zero-communication protocol with $\mu$-randomness for GapInnerProduct (where $\mu$ is any source with maximal correlation $\rho$). It matches the performance of the one-way protocol of \cite{CGMS_ISR}.
\begin{theorem}[zero-communication protocol for GapInnerProduct$_{r,s}^n$]\label{th:smp_gip}
Let $\rho \in (0,1]$ and $-1 \le s < r \le 1$ be given, and let $\mu$ be any source of randomness with maximal correlation $\rho$. There is a zero-communication protocol using $\mu$-randomness that solves GapInnerProduct$_{r,s}^n$ using $O(\frac{1}{\rho^2(r-s)^2})$ bits of communication.
\end{theorem}

We point out that \Cref{th:smp_gip} gives a protocol for \emph{sketching} $\ell_2$-norms using imperfectly shared randomness,  which might be of independent interest. \Cref{th:smp_rest} now follows by combining \Cref{le:cgms_comp} and \Cref{th:smp_gip}. In the rest of this section, we prove \Cref{th:smp_gip}. First, we recall the following observation of \cite{witsenhausen1975sequences} which can be used to convert any source $\mu$ of randomness with maximal correlation $\rho$ to $\bgs(\rho)$.

\begin{proposition}[\cite{witsenhausen1975sequences}]\label{prop:wit}
	Let $\mu$ be a source of randomness with maximal correlation $\rho$. Given access to i.i.d. samples from $\mu$, Alice and Bob can (without interaction) generate i.i.d. samples from $\bgs(\rho)$.
\end{proposition}

\Cref{prop:wit} follows from the definition of maximal correlation and from the two-dimensional Central Limit Theorem. 
We also recall the following well-known fact.
\begin{fact}[Sheppard's formula \cite{sheppard1899application}]\label{fact:sheppard}
If $(X,Y) \sim \bgs(\rho)$ then 
$\Pr[\Sign(X) \neq \Sign(Y)] = \frac{\arccos(\rho)}{\pi}$.
\end{fact}


%
The following lemma is based on the well-known hyperplane rounding technique.
\begin{lemma}\label{le:smp_bgs}
Let $\delta > 0$ and $\gamma < 0$ be given, and let $t = O(\log(1/\delta)/\gamma^2)$ be large enough. Let Alice be given $(X_1,X_2,\dots,X_t) \in \mathbb{R}^t$ and Bob be given $(Y_1,Y_2,\dots,Y_t) \in \mathbb{R}^t$ where $(X_i,Y_i) \sim \bgs(\eta)$ independently over $i \in [t]$. Then, there is a deterministic zero-communication protocol that distinguishes the case where $\eta \geq 0$ from the case where $\eta \le \gamma$ using $O(1/\gamma^2)$ bits of communication, and with probability at least $1-\delta$ (where the probability is over $(X_1,X_2,\dots,X_t)$ and $(Y_1,Y_2,\dots,Y_t)$).
\end{lemma}
\begin{proof}
For every $i \in [t]$, Alice computes $\tilde{X_i} = \Sign(X_i)$ and sends the $t$ bits $\tilde{X}_1, \tilde{X}_2, \dots, \tilde{X}_t$ to the referee. Similarly, Bob computes $\tilde{Y_i} = \Sign(Y_i)$ for each $i \in [t]$, and sends the $t$ bits $\tilde{Y}_1, \tilde{Y}_2, \dots, \tilde{Y}_t$ to the referee. Let $\tau = (\arccos(\gamma)/\pi - 1/2)/2$. The referee computes the Hamming distance $\Delta(\tilde{X}, \tilde{Y})$ and declares that $\eta \geq 0$ if $\Delta(\tilde{X}, \tilde{Y}) \le \tau$, and declares that $\eta \le \gamma$ otherwise. Note that if $\eta \geq 0$, then for each $i \in [t]$,
\begin{equation}\label{eq:disag_ub}
\Pr[\Sign(X_i) \neq \Sign(Y_i)] = \frac{\arccos(\eta)}{\pi} \le \frac{\arccos(0)}{\pi} = \frac{1}{2}.
\end{equation}
On the other hand, if $\eta \le \gamma$, then for each $i \in [t]$,
\begin{equation}\label{eq:disag_lb}
\Pr[\Sign(X_i) \neq \Sign(Y_i)] = \frac{\arccos(\eta)}{\pi} \geq \frac{\arccos(\gamma)}{\pi} = \frac{1}{2} - \Theta(\gamma) - O(\gamma^3),
\end{equation}
where the last equality follows from the Taylor series approximation of $\arccos(x)$ around $x = 0$. The proof now follows by combining \Cref{eq:disag_ub,eq:disag_lb} and an application of the Chernoff bound.
\end{proof}

We are now ready to prove \Cref{th:smp_gip}.

\begin{proof}[Proof of \Cref{th:smp_gip}]

Alice is given $u \in \mathbb{R}^n$ and Bob is given $v \in \mathbb{R}^n$ such that $\|u \|_2 = \| v \|_2 = 1$. They are also given access to i.i.d. samples from a source $\mu$ of randomness with maximal correlation $\rho$. Using \Cref{prop:wit}, Alice and Bob can (without interaction) generate arbitrarily many i.i.d. samples from $\bgs(\rho)$. We first assume that $r = 0$. We will handle the more general case at the end of the proof. Set $\gamma = \rho \cdot s$ and let $t = O(\log(1/\delta)/\gamma^2)$ be as in the statement of \Cref{le:smp_bgs}. 
Draw $t$ i.i.d vectors 
$(X^{(1)},Y^{(1)}), (X^{(2)},Y^{(2)}), \dots, (X^{(t)},Y^{(t)})$ 
each from $\bgs(\rho)^{\otimes n}$. Then, by elliptical symmetry,
we get that independently over $i \in [t]$,
$(\ip{u}{X^{(i)}}, \ip{v}{Y^{(i)}}) \sim \bgs(\rho(\ip{u}{v}))$.
\Cref{le:smp_bgs} now implies a zero-communication protocol that distinguish the case where $\ip{u}{v} \rangle \geq 0$ from the case where $\ip{u}{v} \le s$, using $O(\frac{1}{\rho^2(r-s)^2})$ bits of communication.

We now handle the case where $r$ is not necessarily equal to $0$. First, note that without loss of generality, we can assume that $r \geq 0$. This is because if $r <0$, then Alice can negate each coordinate in her input vector which would preserve its $\ell_2$ norm and replace $r$ by $-s \geq 0$ and $s$ by $-r \geq 0$. Let $N \triangleq n \cdot (1+r)$. Bob will construct a vector $u' \in \mathbb{R}^N$, and Alice will construct a vector $v' \in \mathbb{R}^N$, such that $\|u' \|_2 = \| v' \|_2 = 1$, and:
\begin{enumerate}
\item If $\ip{u}{v} \geq r$, then $\ip{u'}{v'} \geq 0$.
\item If $\ip{u}{v} \le s$, then $\ip{u'}{v'} \le \frac{s-r}{1+r} = -\Theta(r-s)$.
\end{enumerate}
To do so, Alice sets $u'_i  = u_i \cdot \sqrt{n/N}$ for every $i \in [n]$ and $u'_i = +1/\sqrt{N}$ for all $i \in \{n+1,\dots,N\}$. On the other side, Bob sets $v'_i = v_i \cdot \sqrt{n/N}$ for all $i \in [n]$ and $v'_i = -1/\sqrt{N}$ for all $i \in \{n+1,\dots,N\}$.
\end{proof}

\newpage
\section{LSH and Common Randomness}\label{sec:lsh}

An important parameter that governs the performance of an LSH
hash family $\cH$ is given by its $\prho(\cH)$ parameter~\cite{IndykM98}. 
Let $0 \le \alpha \le 1$ and $c \ge 1$.
Loosely speaking, if the hash family ensures 
that points at relative distance at most $\alpha$ collide with 
probability at least $p_1$ while points at relative distance 
at least $c\alpha$ collide with probability at most $p_2$, then 
$\prho(\cH,\alpha,c) \le \log(1/p_1)/\log(1/p_2)$.
Smaller values of $\prho(\cH,\alpha,c)$ can potentially lead
to improvements in the data structure performance.
For the Hamming cube $\zo^d$, there is a trivial scheme $\cH_0$
such that
$\prho(\cH_0) \le \log(1/(1-\alpha))/\log(1/(1-c\alpha)) \to 1/c$
as $\alpha \to 0$.

We show that the zero-communication
common randomness schemes considered here
and in previous 
works~\cite{bogdanov2011extracting,guruswami2016tight}
imply an LSH scheme with
a strictly better $\prho$ parameter. This is perhaps not surprising
since the best strategy for a universal scheme is to map close-by points
to the same output in order to achieve high-agreement probability, 
but to ensure high entropy it must map far-away points to 
different outputs. 

Recall that in the trivial scheme $\cH_0$ the hash function just
outputs the bit at a random coordinate in $[d]$. 
When the relative distance between the two points is $\eps$,
this is tantamount to producing a single sample from $\dsbs(1-2\eps)$.
Thus the trivial LSH scheme is also a trivial common
randomness scheme using one sample from  $\dsbs(1-2\eps)$. 
If we use $k$ samples, i.e. take $k$ independent hash functions,
and use the trivial scheme we obtain an agreement 
$p_{\rho} \ceq \bigl(\tfrac{1+\rho}{2}\bigr)^k$.
Let $f_0(\rho) = \log(1/p_{\rho})/k = \log(2/(1+\rho)$.
In contrast, if we use the mapping given by the common randomness 
scheme then for this hash family (call it $\cH_1$), 
the analogous expression 
equals $f_{cr}(\rho) \ceq (1-\rho)/(1+\rho) + O(\log(k)/k$. 
For large $k$ we can ignore the lower order term.
So let $f_{cr}(\rho) = (1-\rho)/(1+\rho)$. 
To show that $\prho(\cH_2)$ is better we need to show for $\rho > \rho'$ 
that $f_{cr}(\rho)/f_{cr}(\rho') \le f(\rho)/f(\rho')$. That is, 
$f(\rho)/f_{cr}(\rho)$ is an increasing function in $[0,1]$.
This can be verified analytically.
In fact it is always strictly increasing so the bound for the
CR scheme is strictly better than the trivial one.



\end{document}